\documentclass[pra,aps,superscriptaddress,twocolumn,footinbib,notitlepage]{revtex4-2}
\usepackage{mathtools}
\usepackage{amsmath,amssymb,amsfonts}
\usepackage[shortlabels]{enumitem}
\usepackage{graphicx,epic,eepic,epsfig,latexsym,color}
\usepackage{qcircuit}
\usepackage{nicefrac}       
\usepackage[linesnumbered,ruled,procnumbered]{algorithm2e}
\usepackage[table]{xcolor}
\usepackage{framed}
\usepackage{thmtools}
\usepackage{array,multirow,bm,tcolorbox,relsize,booktabs}
\usepackage[utf8]{inputenc}
\usepackage[T1]{fontenc}

\usepackage{bbm}
\usepackage{float}
\usepackage{pgfplots}
\usepackage[strict]{changepage}
\usepackage[marginal]{footmisc}
\usepackage{url}

\usepackage{amsthm}
\theoremstyle{remark}
\newtheorem{definition}{Definition}
\newtheorem{proposition}{Proposition}
\newtheorem{lemma}[proposition]{Lemma}

\newtheorem{theorem}[proposition]{Theorem}
\newtheorem{corollary}[proposition]{Corollary}

\def\squareforqed{\hbox{\rlap{$\sqcap$}$\sqcup$}}
\def\qed{\ifmmode\squareforqed\else{\unskip\nobreak\hfil
\penalty50\hskip1em\null\nobreak\hfil\squareforqed
\parfillskip=0pt\finalhyphendemerits=0\endgraf}\fi}
\def\endenv{\ifmmode\;\else{\unskip\nobreak\hfil
\penalty50\hskip1em\null\nobreak\hfil\;
\parfillskip=0pt\finalhyphendemerits=0\endgraf}\fi}
\newcounter{remark}
\newenvironment{remark}[1][]{\refstepcounter{remark}\par\medskip\noindent%
\textbf{Remark~\theremark #1} }{\medskip}
\newcounter{example}

\newcommand{\nc}{\newcommand}
\nc{\rnc}{\renewcommand}
\nc{\lbar}[1]{\overline{#1}}
\nc{\bra}[1]{\langle#1|}
\nc{\ket}[1]{|#1\rangle}
\nc{\dketbra}[2]{\vert #1 \rangle \hspace{-.8mm} \rangle \hspace{-.4mm} \langle\hspace{-.8mm}\langle #2 \vert}
\nc{\dbra}[1]{\langle\hspace{-.8mm}\langle #1\vert}
\nc{\dket}[1]{\vert#1\rangle\hspace{-.8mm}\rangle}
\nc{\ketbra}[2]{|#1\rangle\!\langle#2|}
\nc{\braket}[2]{\langle#1|#2\rangle}

\nc{\tr}{\operatorname{Tr}}
\nc{\ox}{\otimes}
\nc{\dg}{\dagger}
\nc{\dn}{\downarrow}
\nc{\cA}{{\cal A}}
\nc{\cB}{{\cal B}}
\nc{\cC}{{\cal C}}
\nc{\cD}{{\cal D}}
\nc{\cE}{{\cal E}}
\nc{\cF}{{\cal F}}
\nc{\cG}{{\cal G}}
\nc{\cH}{{\cal H}}
\nc{\cI}{{\cal I}}
\nc{\cJ}{{\cal J}}
\nc{\cK}{{\cal K}}
\nc{\cL}{{\cal L}}
\nc{\cM}{{\cal M}}
\nc{\cN}{{\cal N}}
\nc{\cO}{{\cal O}}
\nc{\cP}{{\cal P}}
\nc{\cQ}{{\cal Q}}
\nc{\cR}{{\cal R}}
\nc{\cS}{{\cal S}}
\nc{\cT}{{\cal T}}
\nc{\cU}{{\cal U}}
\nc{\cV}{{\cal V}}
\nc{\cX}{{\cal X}}
\nc{\cY}{{\cal Y}}
\nc{\cZ}{{\cal Z}}
\nc{\cW}{{\cal W}}
\nc{\csupp}{{\operatorname{csupp}}}
\nc{\qsupp}{{\operatorname{qsupp}}}
\nc{\var}{{\operatorname{var}}}

\nc{\idop}{{\mathbbm{1}}}
\nc{\argmin}{{\operatorname{argmin}}}

\nc{\RR}{{{\mathbb R}}}
\nc{\CC}{{{\mathbb C}}}
\nc{\FF}{{{\mathbb F}}}
\nc{\NN}{{{\mathbb N}}}
\nc{\ZZ}{{{\mathbb Z}}}
\nc{\PP}{{{\mathbb P}}}
\nc{\QQ}{{{\mathbb Q}}}
\nc{\UU}{{{\mathbb U}}}
\nc{\EE}{{{\mathbb E}}}
\nc{\id}{{\operatorname{id}}}

\newcommand{\Choi}{Choi-Jamio\l{}kowski }

\usepackage{hyperref}
\hypersetup{colorlinks=true,citecolor=blue,linkcolor=blue,filecolor=blue,urlcolor=blue,breaklinks=true}
\usepackage{cleveref}

\begin{document}
\title{Optimal Hamiltonian recognition of unknown quantum dynamics}

\author{Chengkai Zhu}\thanks{C.Z. and S.H. contributed equally to this work.}
\author{Shuyu He}\thanks{C.Z. and S.H. contributed equally to this work.}
\author{Yu-Ao Chen}
\author{Lei Zhang}
\author{Xin Wang}
\email{felixxinwang@hkust-gz.edu.cn}
\affiliation{Thrust of Artificial Intelligence, Information Hub,\\
The Hong Kong University of Science and Technology (Guangzhou), Guangzhou 511453, China}

\begin{abstract}
Identifying unknown Hamiltonians from their quantum dynamics is a pivotal challenge in quantum technologies. In this paper, we introduce Hamiltonian recognition, a framework that bridges quantum hypothesis testing and quantum metrology, aiming to identify the Hamiltonian governing quantum dynamics from a known set of Hamiltonians. To identify $H$ for an unknown qubit quantum evolution $\exp(-iH\theta)$ with unknown $\theta$, from two or three orthogonal Hamiltonians, we develop a quantum algorithm for coherent function simulation, built on two quantum signal processing (QSP) structures. It can simultaneously realize a target polynomial based on measurement results regardless of the chosen signal unitary for the QSP. Utilizing semidefinite optimization and group representation theory, we prove that our methods achieve the optimal average success probability, taken over possible Hamiltonians $H$ and parameters $\theta$, decays as $O(1/k)$ with $k$ queries of the unknown unitary transformation. Furthermore, we demonstrate the validity of our protocol on a superconducting quantum processor.
We also investigate a physically motivated recognition task for Heisenberg Hamiltonians, providing numerical evidence for effective multi-qubit quantum system recognition.
This work presents an efficient method to recognize Hamiltonians from limited queries of the dynamics, opening new avenues in composite channel discrimination and quantum metrology.
\end{abstract}

\date{\today}
\maketitle

\section{Introduction}
The evolution of a quantum state in an isolated system is governed by the Schr\"odinger equation, with a Hamiltonian $H$ determining the unitary operation $U_H(\theta) \coloneqq \exp(-i H\theta)$. Learning this unitary operation can be approached through quantum process tomography~\cite{Mohseni2008,Gutoski_2014} or by fitting the parameter $\theta$ based on a hypothesized model $H$, reminiscent of assuming a linear or other model and performing regression in machine learning. In the context of quantum metrology~\cite{Giovannetti_2004,Giovannetti_2006,Degen2017}, the hypothesis on $H$ is fixed and estimating $\theta$ for a known Hamiltonian $H$ has been widely studied and has broad applications in various fields~\cite{Huelga1997,Dorner2009,Anisimov2010,Schnabel_2010,Joo2011,Genoni2011,Thomas2011,Barzanjeh_2015}. The ultimate precision in parameter estimation is a central focus of quantum metrology, with the standard quantum limit scaling as $1/k$ and the Heisenberg limit scaling as $1/k^2$~\cite{Kitagawa1993,Demkowicz2014,Yuan2015,Yuan2016,Pang_2017,Liu_2024}, where $k$ is the number of samples for the target quantum system. Moreover, multi-parameter quantum metrology where the Hamiltonian is not fully known is substantially more complex due to the incompatibility~\cite{Szczykulska_2016,Liu_2019,Demkowicz_Dobrza_ski_2020}. The celebrated quantum Cram\'er-Rao  bound~\cite{helstrom1969quantum,holevo2011probabilistic} can typically not be saturated~\cite{Ragy2016}.

However, could the challenges of multi-parameter estimation be circumvented if we have some prior information about the Hamiltonian? For instance, consider the task of estimating parameters $\alpha$ and $\beta$ encoded in a Hamiltonian $H= \alpha H_0 + \beta H_1$. Suppose we further know there is a possible multiplicative relationship between $\alpha$ and $\beta$, e.g., either $\alpha = a \beta$ or $\alpha = b \beta$ for some known $a,b$. Can we efficiently identify which case actually holds and then reduce the problem to single-parameter metrology? Motivated by this question, we introduce the task of \textit{Hamiltonian recognition}, which is a basic form of Hamiltonian learning.

The specific problem we consider is to guess $H$ from several given candidates by querying the unknown evolution $U_H(\theta)$ multiple times for \textit{arbitrary} $\theta$. The target is to maximize the success probability averaging over all possible $H$ and $\theta$. 
Hamiltonian recognition differs from Hamiltonian learning~\cite{Wang_2017,Bairey2019,Li2020,Che2021,wilde2022,Yu_2023,Huang2023,Gu_2024} or certification~\cite{Wiebe2014}, where the evolution time is typically controllable. In contrast, Hamiltonian recognition aims to develop a time-independent protocol with a single-shot measurement.

The Hamiltonian recognition is closely related to the composite hypothesis testing of quantum channels~\cite{Bergh2023}, where an unknown channel is promised to be from one of two sets, and the task is to decide which set it belongs to by querying it $k$ times. In Hamiltonian recognition, two (or more) sets of unitary operations are generated by different possible Hamiltonians with varying parameters $\theta$. When $\theta$ is fixed, the corresponding binary and multiple unitary discrimination have been extensively studied (e.g., ~\cite{Acin2001,Ariano2001,D_rquot_Ariano_2002,Duan2007,Laing2009,Soeda_2021,Bavaresco_2021,Bavaresco_2022}). It was shown that any pair of unitary operations could be perfectly discriminated with a finite number of uses~\cite{Acin2001,Ariano2001}. Such perfect discrimination may not even require entanglement or joint quantum operations~\cite{Duan2007}. However, little is known about \textit{composite unitary discrimination} where the goal is to distinguish between two sets of unitary operations, potentially containing an infinite number of elements. The bottleneck arises from the diverse strategies that can be applied to unitary operations, such as parallel, sequential, and indefinite causal order~\cite{Oreshkov_2012,Chiribella2013,Bavaresco_2021}, as well as the complexity arising from two sets of unitaries.

In this work, we develop a protocol for binary recognition of the underlying single-qubit Hamiltonian, given as an unknown oracle $\exp(-iH\theta)$. This protocol achieves an average error of $O(1/k)$ where $k$ is the number of queries to the unknown unitary, a scaling that is proven to be optimal when taking the average over all possible Hamiltonians and parameters. Remarkably, our protocol utilizes the technique of quantum signal processing (QSP)~\cite{Low_2016,low_optimal_2016} to analyze protocol resolution errors. The QSP provides a systematic approach to making a quantum system transform under certain polynomial functions~\cite{yu_power_2022} and is also helpful for analyzing multiple unitary discrimination~\cite{Rossi2021}. As a byproduct, we reveal that for any positive integer $k$, there exist two sets, each containing $k$ qubit unitaries, such that many pairs of unitaries, one from each set, are not orthogonal, yet these sets can be perfectly distinguished using only $k$ queries without the need for entanglement.

Notably, we demonstrate the optimality of our protocol by using semidefinite programming (SDP). We also show different strategies involving processes with indefinite causal order~\cite{Oreshkov_2012,Chiribella2013}, will not offer any advantage in the recognition of Pauli $X$ and Pauli $Z$. To validate our results, we implement our binary recognition protocol on a superconducting quantum processor from Tencent and observe the desired error decay. Furthermore, to tackle ternary Hamiltonian recognition, we develop a quantum algorithm, consisting of coherent uses of two QSP circuits, that can realize a target polynomial based on the measurement results regardless of which signal unitary is chosen for the QSP. This protocol achieves an optimal average error of $O(1/k)$ for odd $k$.

\section{Results}
\subsection{Preliminaries}
To address the task of Hamiltonian recognition, we begin by introducing quantum channel discrimination~\cite{Pirandola_2019,Wilde_2020,Zhuang2020,Li2022,Salek2022}. Let $\cL(\cH)$ be the set of linear operators acting on a finite-dimensional Hilbert space $\cH$. Quantum channels from an input system $I$ to an output system $O$ are completely positive and trace-preserving maps $\cN:\cL(\cH_{I})\rightarrow \cL(\cH_{O})$. The Choi operator of a quantum channel is defined as $J^{\cN}_{IO}\coloneqq \sum_{j, l=0}^{d_I-1}\ketbra{j}{l}_{I} \ox \cN_{I \to O}(\ketbra{j}{l}_{I})$. Particularly, the Choi operator of a unitary $U$ is denoted as $\dketbra{U}{U}$ where $\dket{U}\coloneqq \sum_j (\idop\ox U)\ket{jj}$. Given access to one of two possible quantum channels, $\cN_0$ or $\cN_1$, the objective is to determine which channel has been provided by querying the unknown channel. This is typically done by sending a quantum state into the channel and performing a positive operator-valued measure (POVM), given by a tuple $M\coloneqq \{M_0, M_1\}$ where $M_1 = \idop - M_0$. If the measurement outcome is `0', one guesses the channel is $\cN_0$; if the outcome is `1', one guesses the channel is $\cN_1$.

The notion of a \textit{tester} has been proposed as a useful tool to study quantum channel discrimination~\cite{Chiribella_2008,Chiribella_2008tester,Bavaresco_2021,Bavaresco_2022}. A binary tester is a set of positive semidefinite operators $T \coloneqq \{T_0, T_1\}$ that obey certain linear constraints. By taking the trace with the Choi operator of a quantum channel $J^{\cN}_{IO}$, it results in a valid probability distribution according to $p(\cN_{I \to O}|J^{\cN}_{I O}) = \tr(T_0 J^{\cN}_{I O})$. Essentially, a tester is a quantum comb without the dimension for input and output systems, simplifying the study of quantum channel discrimination by encompassing any possible scheme with parallel, sequential, and combined setups of the tested channels, depending on the linear constraints imposed on $W\coloneqq\sum_jT_j$. The set of all testers with different strategies is denoted as $\cT^{\cS}$ with $\cS\in\{\texttt{PAR}, \texttt{SEQ}, \texttt{GEN}\}$, which correspond to parallel, sequential, and general strategies, respectively. The optimization over quantum testers essentially optimizes over all possible input probes, intermediate operations, and final measurements. Further details about the tester formalism can be found in~\cite{Bavaresco_2021} and \footnote{See Supplementary Information}.

\begin{figure*}[t]
    \centering
    \includegraphics[width=.95\linewidth]{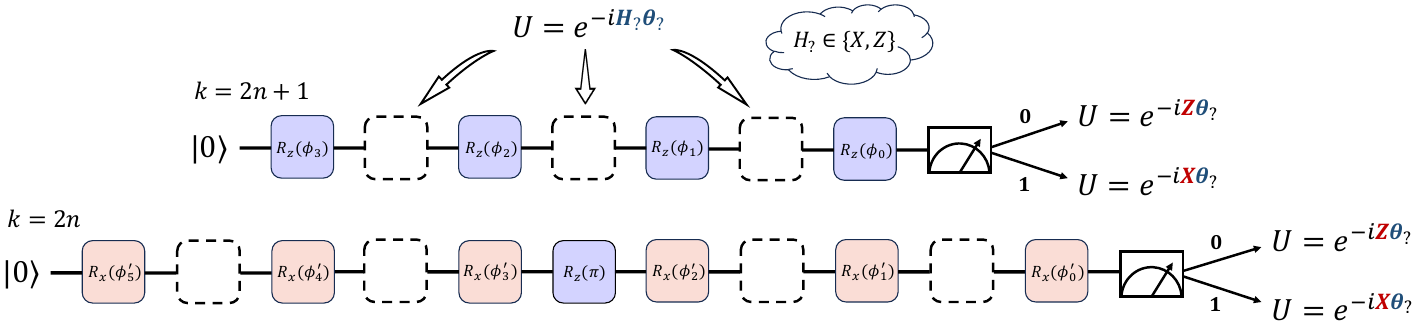}
    \caption{Optimal protocols for the recognition of $\{X,Z\}$. The protocols consist of rotation gates around the $x$-axis (in pink) and the $z$-axis (in purple). The top circuit illustrates the protocol for an odd number of queries, while the bottom circuit represents the protocol for an even number of queries.}
    \label{fig:figure1}
\end{figure*}

\subsection{Binary Hamiltonian recognition}
We now formally introduce the problem of Hamiltonian recognition. Suppose we have access to a unitary operation $U_H\coloneqq\exp(-iH\theta)$ generated by an unknown Hamiltonian $H$, where $H$ belongs to a given known set $\{H_0, H_1, \cdots, H_{m-1}\}$ and $\theta$ is an unknown parameter. The task is to identify the generating Hamiltonian $H$ through multiple queries of the unknown evolution, without knowing the parameter $\theta$. We aim to develop a protocol that maximizes the average success probability, taking over all possible Hamiltonians in the given set and all $\theta$ within a specified range. To formalize this problem mathematically, we introduce the optimal average success probability of Hamiltonian recognition as follows.

\begin{definition} (Hamiltonian Recognition)
For a set of $m$ Hamiltonians $\{H_j\}_{j=0}^{m-1}$, let $p_j$ be the prior probability that $H_j$ is given where $\sum_{j=0}^{m-1} p_j = 1$, and denote $U_{H_j}(\theta_j) = e^{-i H_j \theta_j}$ as the unitary operation generated by $H_j$ where $\theta_j \in [0, t]$. The optimal average success probability of recognizing $\{H_j\}$ by $k$ queries with strategy $\cS$ is defined as
\begin{equation*}
\begin{aligned}
    {\rm Suc}_k^{\cS}(\{H_j\}, \{p_j\}, t) &= \max_{\{T_j\}} \frac{1}{t^m}\int_{0}^{t} \cdots \int_{0}^{t} {\rm d} \theta_0 \cdots {\rm d} \theta_{m-1}\\
    &\sum_{j=0}^{m-1} p_j \tr\Big(T_j\dketbra{U_{H_j}(\theta_j)}{U_{H_j}(\theta_j)}^{\ox k}\Big),
\end{aligned}
\end{equation*}
where the maximization ranges over all $k$-slot testers $\{T_j\}_{j=0}^{m-1} \in \cT^{\cS}$ and $\cS\in \{\texttt{PAR}, \texttt{SEQ}, \texttt{GEN}\}$.
\end{definition}
Physically, this definition corresponds to a strategy where an entangled input state evolves under $k$ uses of the unitary and is measured by POVM elements $M_j$. Consequently, the element of a tester $\{T_j\}$ mathematically encapsulates the entire setup, representing both the specific input state and the POVM measurement~\cite{Chiribella_2008tester}. See Supplementary Information for detailed characterizations of parallel, sequential, and general testers, which describe the different methods for querying unknown dynamics.

The integration over the parameter $\theta$ corresponds to assuming a uniform prior distribution, where all possible values are considered equally likely for each type of Hamiltonian. This ensures the quantum tester $\{T_j\}$ is independent of the unknown parameter.

For $k$ copies of a unitary operation, the average Choi operator of unitary operations with different parameters $\theta$ can be written as 
\begin{equation}
\Omega^{(k)}_{H}(t) = \frac{1}{t}\int_0^{t} \dketbra{U_H(\theta)}{U_H(\theta)}^{\ox k} \mathrm{d} \theta,
\end{equation}
which is also called the \textit{performance operator}~\cite{Chiribella_2016}. We can then rewrite ${\rm Suc}_k^{\cS}(\{H_j\}, \{p_j\}, t)$ as
\begin{equation}\label{Eq:opt_suc_pri}
    {\rm Suc}_k^{\cS}\big(\{H_j\}, \{p_j\}, t\big) = \max_{\{T_j\}} \sum_{j=0}^{m-1} p_j \tr\Big(T_j\Omega_{H_j}^{(k)}(t)\Big).
\end{equation}

For the case of $m=2$, we start with binary Hamiltonian recognition, where the unknown Hamiltonian $H$ is selected from the set $\{H_0, H_1\}$. Let $\vec{n} \coloneqq (n_x, n_y, n_z)$ be a real unit vector in three dimensions and $H = \vec{n}\cdot \vec{\sigma}$ be a single-qubit Hamiltonian where $\vec{\sigma} = (X, Y, Z)$, and $X = \left(\begin{array}{cc}
0 & 1 \\
1 & 0
\end{array}\right),~
Y = \left(\begin{array}{cc}
0 & -i \\
i & 0
\end{array}\right),~
Z = \left(\begin{array}{cc}
1 & 0 \\
0 & -1
\end{array}\right)$ are single-qubit Pauli operators. Then $e^{-i \theta \vec{n}\cdot \vec{\sigma}}$ is a rotation by $\theta$ about the $\vec{n}$ axis. Note that we can always normalize Hamiltonians first and absorb a normalization factor into the parameter $\theta$. Our first main result characterizes the optimal binary Hamiltonians recognition of $H_0$ and $H_1$, whose $\vec{n}_0$ and $\vec{n}_1$ are orthogonal, using $k$ queries to the unknown qubit evolution.

\begin{figure*}[t]
    \centering
    \includegraphics[width=1\linewidth]{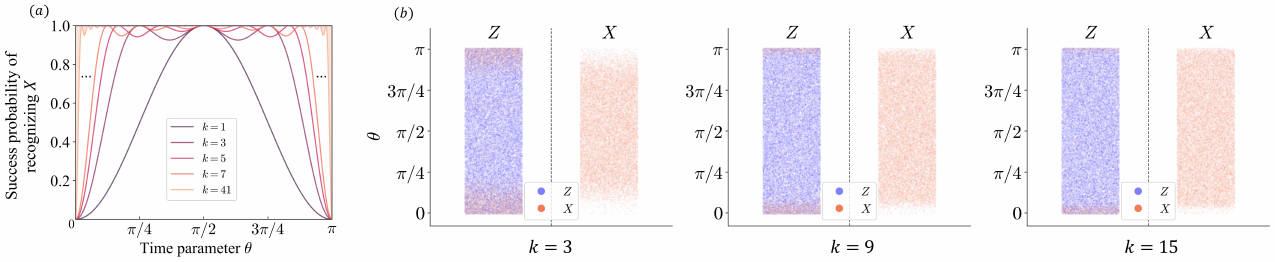}
    \caption{Performance of binary Hamiltonian recognition for ${X, Z}$. (a) The success probability of recognizing $X$ by querying an unknown unitary $k$ times. The $x$-axis represents the parameter $\theta$ for the unknown unitary, and the $y$-axis corresponds to the success probability of recognizing $X$. (b) Experiment results of the recognition protocol on a surface-13 tunable coupling superconducting quantum processor with different numbers of queries to the unknown unitary. In each panel, the $y$-axis corresponds to $\theta$ for $e^{-iH\theta}$, where the blue points represent unitaries generated by the Hamiltonian $H=Z$ and the orange points represent unitaries generated by $H=X$. For each value of $\theta$, 512 points of each color are shown, corresponding to 1024 repeated experiments in total. The left region corresponds to unitaries generated by $Z$, and the right region to those generated by $X$. Increasing the number of queries reduces the misclassification of points between the two regions.}
    \label{fig:bi_Xerr}
\end{figure*}

\begin{theorem}\label{thm:recog_XZ}
For $H_0 = \vec{n}_0\cdot \vec{\sigma}$ and $H_1 = \vec{n}_1\cdot \vec{\sigma}$ with prior probabilities $p$ and $1-p$, respectively, if $\vec{n}_0 \perp \vec{n}_1$, the optimal average success probability of recognizing $\{H_0,H_1\}$ using $k$ queries is given by
\begin{equation*}
{\rm Suc}_k^{\texttt{GEN}}\left(\{H_0, H_1\},  \{ p, 1-p\}, \pi \right) = \frac{k+\max \{p,1-p\}}{k+1}.
\end{equation*}
This limit is globally optimal, holding for all $k$-query strategies, and is saturated by a sequential protocol without entanglement.
\end{theorem}

We first prove the result for $H_0 = X$ and $H_1 = Z$ and reduce the general cases to this one. Notice that these rotation gates are periodic with respect to $\theta$, so we can consider the case where $\theta\in(0,\pi]$. To recognize a Hamiltonian by querying the underlying unitary operation, intuitively, we can choose the probe and measurement basis as its eigenstate, e.g., $\ket{0}$ for $\exp(-iZ\theta)$, to make a perfect recognition. Then would the error of recognizing $\exp(-iX\theta)$ be some functions of $\theta$? Following this intuition, we construct a recognition protocol utilizing the framework of QSP~\cite{Low_2016,low_optimal_2016} which has been widely utilized for the unification of well-known quantum algorithms~\cite{gilyen_quantum_2018, Martyn_2021}, the development of new quantum protocols~\cite{yu_power_2022,Wang_2023,low2019hamiltonian,lloyd2021hamiltonian,martyn2022efficient,childs2018toward,rall2021faster, gur2021sublinear,li2019quantum,subramanian2021quantum,tong2021fast,gribling2024an, odake2024higherorder}. A representative QSP protocol involves applying a gate sequence of the form
\begin{equation}
e^{iZ\phi_0} e^{iX\theta}e^{iZ\phi_1} e^{iX\theta}e^{iZ\phi_2} \cdot \ldots \cdot e^{iX\theta}e^{iZ\phi_k},
\end{equation}
where $e^{iX\theta}$ is referred to as the signal unitary with an unknown parameter $\theta$, and the angles $\phi_0, \phi_1, \ldots, \phi_k$ are controllable and referred to as the \textit{QSP phase sequence}.

In specific, Fig.~\ref{fig:figure1} illustrates the protocols for querying odd and even times. The odd case is described as follows, and the details of the protocol and proof can be found in the Supplementary Information. We explicitly chose $k+1$ angles $\phi_j,j=0,1,...,k$. Given an unknown unitary $e^{-iH\theta}$, we sequentially apply $e^{iZ\phi_k} e^{-iH\theta} e^{iZ\phi_{k-1}} e^{-iH\theta}\cdots e^{iZ\phi_1} e^{-iH\theta} e^{iZ\phi_0}$ to the zero state. Then we measure the state on a computational basis. If outcome is $0$, we decide $H=Z$; if outcome is $1$, we decide $H=X$. Denote the composed unitary as
\begin{equation}
    {\rm QSP}_{H,k}(\theta) \coloneqq e^{iZ\phi_0} \prod_{j=1}^k e^{-iH\theta}  e^{iZ\phi_j},
\end{equation}
It is easy to see that the outcome is always $0$ if $U_Z(\theta)$ is given, and the error of determining $H=Z$ when $X$ is given by $f_k(\theta) \coloneqq |\bra{0}{\rm QSP}_{X,k}(\theta)\ket{0}|^2$. To this end, we analyze what polynomial functions $f_k(\theta)$ can be realized by a QSP phase sequence $\vec{\phi}=(\phi_0,\phi_1,\cdots,\phi_k)$ to study the average recognition error. We show that there exists a QSP phase sequence that achieves the following polynomial, which is also related to the threshold function~\cite{gilyen_quantum_2018,martyn_grand_2021}.
\begin{equation}\label{Eq:threhold_func}
    f_k(\theta) = \sum_{l=-k}^{k} \frac{k-|l|+1}{(k+1)^2} e^{2il\theta}.
\end{equation}
This yields an average error $1/(k+1)$ for $X$ after averaging over all $\theta\in (0,\pi]$. Overall, the average success probability over $X,Z$ and $\theta$ is $(k+\max \{p,1-p\})/(k+1)$ where $p$ and $1-p$ are prior probability of the Hamiltonian being $X$ and $Z$, respectively.

Notably, we prove that the protocol constructed via QSP achieves optimality in this Hamiltonian recognition task. To understand the ultimate limit of all possible protocols for any number of queries, we delve into the SDP formalism of the problem. Followed by Eq.~\eqref{Eq:opt_suc_pri}, the optimal average success probability of recognizing $\{X,Z\}$ can be written as
\begin{equation}\label{Eq:opt_suc_pri_XZ}
    \max_{\{T_0, T_1\}\in \cT^{\texttt{SEQ}}}\; \tr\Big(pT_0 \Omega_{X}^{(k)}(\pi) + (1-p) T_1 \Omega_{Z}^{(k)}(\pi)\Big).
\end{equation}
Drawing upon the method introduced in Ref.~\cite{Bavaresco_2021} for deriving the dual problem by characterizing dual affine spaces, we formulate the dual problem of Eq.~\eqref{Eq:opt_suc_pri_XZ} as follows, ensuring that strong duality is satisfied~\cite{boyd2004convex}.
\begin{equation}\label{sdp:XZ}
\begin{aligned}
    {\rm Suc}_k^{\texttt{SEQ}}\left(\{X, Z\}, \{p,1-p\},\pi \right) 
    = \min &\;\; \lambda \\
    {\rm s.t.} &\;\; \Omega_{X}^{(k)}(\pi) \leq \frac{1}{p}\lambda \overline{W},\\
    &\;\; \Omega_{Z}^{(k)}(\pi) \leq \frac{1}{1-p}\lambda\overline{W},\\
    &\;\; \overline{W} \in \overline{\cW}^{\texttt{SEQ}}.
\end{aligned}
\end{equation}
We analytically construct feasible solutions $\{\lambda_k, \overline{W}_k\}$ for any $k$, by using group representation theory to showcase the matrix inequality in this SDP problem. Then, it can be shown that the optimal average success probability is upper bounded by $(2k+1)/(2k+2)$. Furthermore, we show that general strategies, such as those involving the quantum switch~\cite{Oreshkov_2012,Chiribella2013}, cannot improve this upper bound. Our QSP-based protocol saturates this upper bound and thus achieves the optimal average success probability considering any query model and input probes. Detailed proofs can be found in the Supplementary Information.

Theorem~\ref{thm:recog_XZ} has several remarkable implications. First, it demonstrates that, given $k$ uses of the unknown evolution, the optimal average recognition error decays efficiently as $O(1/k)$. Notably, the optimal protocol exhibits no error when the unknown Hamiltonian is $Z$, while the error in deciding the unknown Hamiltonian is $Z$ when it is actually $X$ is $1/(k+1)$. This result also provides an operational interpretation of the polynomial in Eq.~\eqref{Eq:threhold_func}, generated by a $ZXZ$-type QSP.

In Fig.~\ref{fig:bi_Xerr} (a), we illustrate the success probability of determining $X$ for different $k$, where the $x$-axis represents the parameter $\theta$ for $e^{-iX\theta}$ and the $y$-axis corresponds to the polynomial $1-f_k(\theta)$ generated by the QSP circuit. The shaded area corresponds to the average error and vanishes very quickly. We further note that the angles $\theta_j= \frac{j\pi}{k+1}, j=1,2,\cdots, k$ are the zeros of the polynomial $1-f_k(\theta)$, which implies that the error probability of recognizing $X$ vanishes at these points. Consequently, the two sets of unitary operations $\{e^{-iX\theta_j}\}_{j=1}^k$ and $\{e^{-iZ\theta_j}\}_{j=1}^k$ can be \textit{perfectly} distinguished via our protocol.
\begin{corollary}\label{cor:XZ_perfect_dis}
    For any $k\in \NN^+$, there exist two disjoint sets of qubit unitary operations $\Theta_0\coloneqq\{U_j\}_{j=0}^{k-1}$ and $\Theta_1\coloneqq\{V_j\}_{j=0}^{k-1}$ such that $\Theta_0\cap \Theta_1 = \emptyset$, which can be perfectly distinguished by querying the unknown unitary adaptively $k$ times.
\end{corollary}
When $k$ is even, we can construct $\Theta_0$ and $\Theta_1$ such that any pair of unitaries, one from each set, are not orthogonal. This simple example extends our understanding of unitary discrimination, going beyond the previous knowledge that any two unitary operations can be perfectly distinguished with~\cite{Ariano2001,Acin2001}, or without using entanglement~\cite{Duan2007}. We now provide an example that even two sets of unitary operations, whose elements could be highly mutually nonorthogonal, can be perfectly distinguished with a finite number of uses, without the need for entanglement.

Further, our result nicely relates the QSP to the problem of quantum Hamiltonian recognition. Our protocol, consisting of a sequence of rotation gates, is easy to implement and achieves the optimal success probability. Moreover, the protocol is independent of the unitary evolution time. To the best of our knowledge, it is the first concrete optimal protocol for composite hypothesis testing of unitary operations. Notably, as the optimal protocol requires only single-qubit quantum gates, we establish that an entangled probe provides no advantage in this task.

\begin{figure*}[t]
\centering
    \includegraphics[width=1\linewidth]{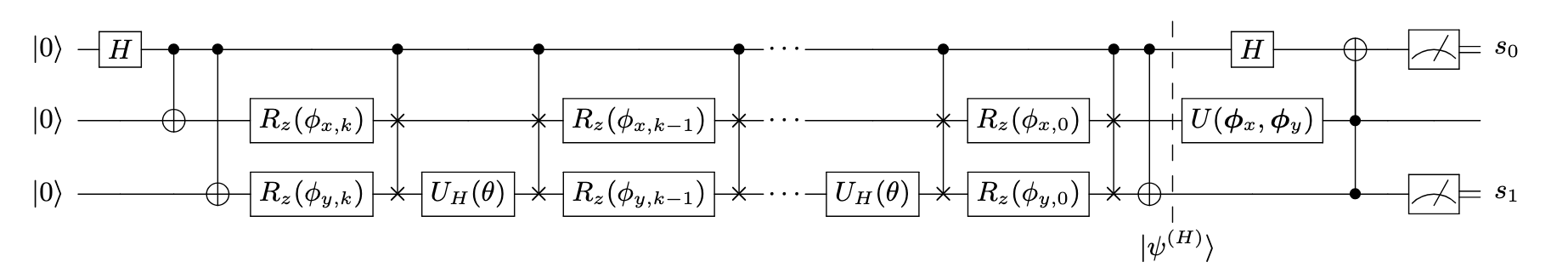}
\caption{The three-qubit quantum circuit for ternary Hamiltonian recognition of $\{X,Y,Z\}$. The input state is $\ket{000}$. If the measurement outcome is $s_0s_1 = 00$ or $10$, determine $H=Z$. If the measurement outcome is $s_0s_1 = 01$, determine $H=Y$. If the outcome is $s_0s_1 = 11$, determine $H=X$.}
\label{fig:XYZ_qsp}
\end{figure*}

To validate our theoretical result, we implemented the proposed protocol on a surface-13 tunable coupling superconducting quantum processor~\cite{Bu2024}. For this processor~\cite{Bu2024}, the median single-qubit gate fidelity shows a tiny decrease from about 99.95\% for the isolated Randomized Benchmarking to 99.94\% for the simultaneous one. 
We uniformly sample 64 values of $\theta_j$ from $(0,\pi]$ to construct two sets of gates, $\{e^{-iX\theta_j}\}_{j=1}^k$ and $\{e^{-iZ\theta_j}\}_{j=1}^k$, corresponding to the Hamiltonians $X$ and $Z$, respectively. For each gate, our protocol was executed 512 times, and we recorded our predictions based on the measurement outcomes. The results are illustrated in Fig.~\ref{fig:bi_Xerr} (b), where blue points represent $e^{-iZ\theta_j}$ and orange points represent $e^{-iX\theta_j}$. For each angle on the $y$-axis, the unitary is placed in the left region if we predict $Z$ and in the right region if we predict $X$. We present our 512 prediction outcomes for each angle and different $k$. The experimental findings highlight the practicality and robustness of our protocol when applied to a real-world quantum computing platform. The implementation details can be found in the Supplementary Information.

\subsection{Ternary Hamiltonian recognition}
Apart from the binary recognition, we further explore the scenario of ternary Hamiltonian recognition, where the unknown Hamiltonian $H$ is selected from a set $\{H_0, H_1, H_2\}$ with prior probabilities $p_0,p_1,p_2$, respectively, where $p_0+p_1+p_2=1$.

\begin{theorem}~\label{thm:recog_XYZ}
For $H_0 = \vec{n}_0\cdot \vec{\sigma},~H_1 = \vec{n}_1\cdot \vec{\sigma}$ and $H_2 = \vec{n}_2\cdot \vec{\sigma}$ with prior probabilities $p_0,p_1,p_2$, respectively, where $p_0+p_1+p_2=1$, if $\vec{n}_0 \perp \vec{n}_1 \perp \vec{n}_2$, the optimal average success probability of recognizing $\{H_0, H_1, H_2\}$ via a $k$-slot sequential strategy, where $k$ is odd, is given by
\begin{equation}
    {\rm Suc}_k^{\texttt{SEQ}}\left(\big\{H_i\big\}_{i=0}^2, \{p_i\}_{i=0}^2, \pi \right) = \frac{k+\max\{p_0,p_1,p_2\}}{k+1}.
\end{equation}
\end{theorem}

The result can also be understood from the specific case of $H_0=X, H_1=Y, H_2=Z$. For achievability, the key idea is to utilize the QSP protocol in Theorem~\ref{thm:recog_XZ} for $\{X, Z\}$ and $\{Y,Z\}$ coherently. We denote the QSP circuit for recognizing $\{X, Z\}$ as
\begin{equation}
    {\rm QSP}^{\{X,Z\}}_{H,k}(\theta) = e^{iZ\phi_{x,0}} \prod_{j=1}^k e^{-iH\theta}  e^{iZ\phi_{x,j}},
\end{equation}
We denote ${\rm QSP}^{\{X,Z\}}_{X,k}(\theta)\ket{0} = P_{x,k}(\theta) \ket{0} + Q_{x,k}(\theta) \ket{1}$ and ${\rm QSP}^{\{X,Z\}}_{X,k}(\theta)\ket{1} = -Q_{x,k}^*(\theta) \ket{0} + P_{x,k}^*(\theta) \ket{1}$ with $f_{k}(\theta) = |P_{x,k}(\theta)|^2$. Similarly, we denote the QSP circuit for recognizing $\{Y, Z\}$ as
\begin{equation}
    {\rm QSP}^{\{Y,Z\}}_{H,k}(\theta) = e^{iZ\phi_{y,0}} \prod_{j=1}^k e^{-iH\theta}  e^{iZ\phi_{y,j}},
\end{equation}
and ${\rm QSP}^{\{Y,Z\}}_{Y,k}(\theta)\ket{0} = P_{y,k}(\theta) \ket{0} + Q_{y,k}(\theta) \ket{1}$. In particular, we could set $Q_{y,k}(\theta) = iQ_{x,k}^*(\theta)$. 

Based on the above, we developed a quantum algorithm that can simultaneously realize the polynomial $f_k(\theta)$ regardless of signal unitary is $e^{-iX\theta}$ or $e^{-iY\theta}$. The quantum circuit of the algorithm is depicted in Fig.~\ref{fig:XYZ_qsp}. We input the initial state $\ket{000}$ into the circuit, apply the unknown unitary $U_H$ interleaved with other predetermined gates, and measure the first and third qubits of the output state on a computational basis. Denote the measurement outcome as $s=s_0s_1$. If $s=00$ or $10$, decide $H = Z$. If $s=01$, decide $H = Y$. If $s=11$, decide $H = X$.

If $H=Z$, again, there is no error as we will always have $s_1 = 0$. If $H = X$, we focus on the state $\ket{\psi^{(X)}}$ shown in Fig.~\ref{fig:XYZ_qsp}. The branch whose third qubit is in $\ket{1}$ can be written as
\begin{equation}
\begin{aligned}
    &\left(H\ox \idop \ox \ketbra{1}{1}\right) \ket{\psi^{(X)}} \\
    =&\; Q_{x,k}^* (\theta) \Bigg[\ket{0} \ox \left(\frac{e^{-i\boldsymbol{\phi}_x/2}}{2} \ket{0} - \frac{e^{i\boldsymbol{\phi}_y/2}}{2}\ket{1} \right)\ox \ket{1}\\
    &\qquad + \ket{1} \ox \left(\frac{e^{-i\boldsymbol{\phi}_x/2}}{2} \ket{0} + \frac{e^{i\boldsymbol{\phi}_y/2}}{2}\ket{1} \right)\ox \ket{1}\Bigg],
\end{aligned}
\end{equation}
where $\boldsymbol{\phi}_l \coloneqq \sum_j \phi_{l,j},l\in\{x,y\}$. Choosing
\begin{equation}
U(\boldsymbol{\phi}_x, \boldsymbol{\phi}_y) = \frac{1}{\sqrt{2}}\left(
\begin{array}{cc}
   e^{i\boldsymbol{\phi}_y/2}  & e^{-i\boldsymbol{\phi}_x/2} \\
   e^{i\boldsymbol{\phi}_x/2}  & -e^{-i\boldsymbol{\phi}_y/2}
\end{array}\right),
\end{equation}
the final state of the circuit, given that the third qubit is in $\ket{1}$, is 
\begin{equation*}
    Q_{x,k}^*(\theta) \Bigg[\frac{1}{\sqrt{2}} \ket{111} + \frac{e^{-i(\boldsymbol{\phi}_x - \boldsymbol{\phi}_y)/2}}{\sqrt{2}} \ket{101}\Bigg].
\end{equation*}
The success probability of recognizing $X$ (measured `11') is given by $|Q_{x,k}^*|^2$. Similar results could be obtained for the case $H=Y$. After calculation for each case, we can arrive at an average success probability $(3k+1)/(3k+3)$. Further, the optimality of our protocol among all possible sequential or general strategies is proved via the dual SDP formulation. All detailed proofs can be found in the Supplementary Information.

We note that different from directly using two QSP protocols coherently, which would stop at $\ket{\psi^{(H)}}$ in the circuit and perform a measurement, our key construction after $\ket{\psi^{(H)}}$ amplifies the success probability of recognizing $X$ (or $Y$) from $|Q_{x,k}^*|^2/2$ to $|Q_{x,k}^*|^2$. Because the error polynomials w.r.t $\theta$ for recognizing $X$ and $Y$ are the same, i.e., $1-f_k(\theta)$, similar to Corollary~\ref{cor:XZ_perfect_dis}, we can construct three pairwise disjoint sets of $k$ unitary operations that can be perfectly distinguished by querying $k$ times via our protocol.

\begin{figure}[t]
    \centering
    \includegraphics[width=1\linewidth]{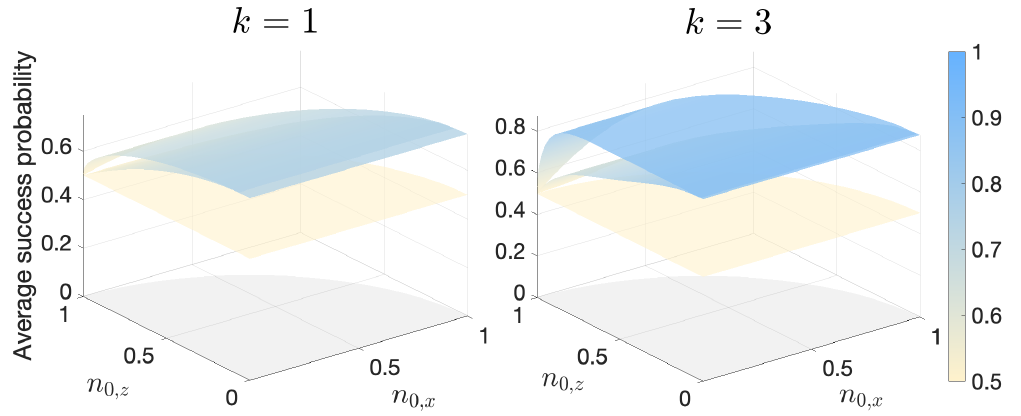}
    \caption{Average success probability of recognizing $\{H_0, Z\}$. The plots illustrate the performance for query counts $k=1$ (left) and $k=3$ (right) as a function of the vector components $n_{0,x}$ and $n_{0,z}$ of the Hamiltonian $H_0$. In each panel, three performance layers are compared: the top surface represents the globally optimal success probability derived via semidefinite programming (SDP); the middle surface depicts the performance of the proposed ${X, Z}$ QSP protocol applied to this general case; and the bottom surface represents the random guessing baseline.}
    \label{fig:figure4}
\end{figure}

\subsection{General axis directions and multi-qubit Hamiltonian}
Now we consider a more general scenario where $\vec{n}_0$ and $\vec{n}_1$ are not orthogonal for $H_0 = \vec{n}_0\cdot\vec{\sigma}$ and $H_1 = \vec{n}_1\cdot\vec{\sigma}$. Without loss of generality, we set $H_1 = Z$. Since $\vec{n}_0\coloneqq(n_{0,x},n_{0,y},n_{0,z})$ is a unit vector, we can find a unitary $U$ that diagonalize $H_0$ such that $UH_0U^\dagger = Z$. Therefore, for fixed $k$ and $\vec{n}_0$, we could obtain $\Omega_{H_0}^{(k)}(\pi) = (U\ox U^*)\Omega_{Z}^{(k)}(\pi)(U^\dagger \ox U^T)$. Further, we numerically solve SDP~\eqref{sdp:XZ} to estimate ${\rm Suc}_k^{\texttt{SEQ}}\left(\big\{H_0, Z\big\}, \pi \right)$ for varying $n_{0,x}$ and $n_{0,z}$.

We illustrate the average success probability of recognizing $\{H_0, Z\}$ in Fig.~\ref{fig:figure4}, for the cases of $k=1$ and $k=3$. In each subplot, the top layer represents the average success probability corresponding to the optimal protocol obtained by solving the SDP; the middle layer depicts the average success probability given by the protocol in Theorem~\ref{thm:recog_XZ} specifically designed for $\{X, Z\}$, and the bottom layer shows the performance of random guessing. It can be observed that within a certain range, the $\{X, Z\}$ protocol performs well. As the number of slots increases, the optimal protocol exhibits high accuracy over a larger range.

Furthermore, we consider a more physically motivated recognition task. We aim to distinguish between two equiprobable Heisenberg-type Hamiltonians, $\hat{H}_0$ and $\hat{H}_1$, of the form 
\begin{equation}
H(J_x,J_y) = \sum_{j=1}^N J_x X_j X_{j+1} + J_y Y_j Y_{j+1}    
\end{equation}
with real-valued coupling constants $J_x, J_y$. For each instance of the problem, we define the pair to be distinguished as $\hat{H}_0 = H(p,p)$ and $\hat{H}_1 = H(-p,-p)$ for a given parameter $p$.

We approach this problem using a variational method. The $k$-query recognition protocol is modeled by a parameterized quantum circuit~\cite{mo2025parameterized}, $C_b(\bm \phi, \theta) = U(\bm \phi^{(k)})\prod_{i=1}^k e^{-i\hat{H}_b\theta_i} U(\bm \phi^{(i-1)})$, where the evolutions under the unknown Hamiltonian are interleaved with trainable universal two-qubit ansatzes $U(\bm \phi^{(i)})$ when $N=2$ (universal three-qubit ansatzes when $N=3$). The set of all gate parameters, ${\bm \phi^{(0)}, ..., \bm \phi^{(k)}}$, constitutes the variational parameters $\bm \phi$ of our protocol. The optimal protocol is found by minimizing the average failure probability. Since the evolution times $\theta_i$ are unknown parameters, we optimize for the averaged scenario, with a cost function given by
\begin{equation*}
\cL(\bm \phi) = \frac{1}{2M^2}\sum_{i=1}^M\sum_{j=1}^M \left( \big| \bra{1}C_0(\bm \phi, \theta_i)\ket{0} \big|^2 + \big| \bra{0}C_1(\bm \phi, \theta_j)\ket{0} \big|^2 \right)
\end{equation*}
where ${\theta_i}$ and ${\theta_j}$ are sets of $M$ evolution times sampled uniformly from a range $[0, t]$.
We use a gradient-based optimizer, e.g., Adam, to find the optimal parameters $\bm \phi^* = \arg\min_{\bm \phi} \cL(\bm \phi)$.

Fig.~\ref{fig:heisen_2_3qubit} illustrates the performance of our optimized protocol for recognizing $\{\hat{H}_0, \hat{H}_1\}$ with $N=2$ and $N=3$, where we choose $t=\pi, M=3142$ during our optimization. We consistently observe that the success probability increases with the number of queries, $k$, for all instances of Hamiltonian pairs. This provides numerical evidence that effective protocols for recognizing different Heisenberg Hamiltonians exist. However, analog from single-qubit QSP to its multi-qubit extensions~\cite{gilyen_quantum_2018, Wang_2023}, it is non-trivial to extend our single-qubit results to high-dimensional cases. We expect more systematic efforts in future work to establish formal guarantees for multi-qubit and many-body systems.

\begin{figure}[t]
    \centering
    \includegraphics[width=1\linewidth]{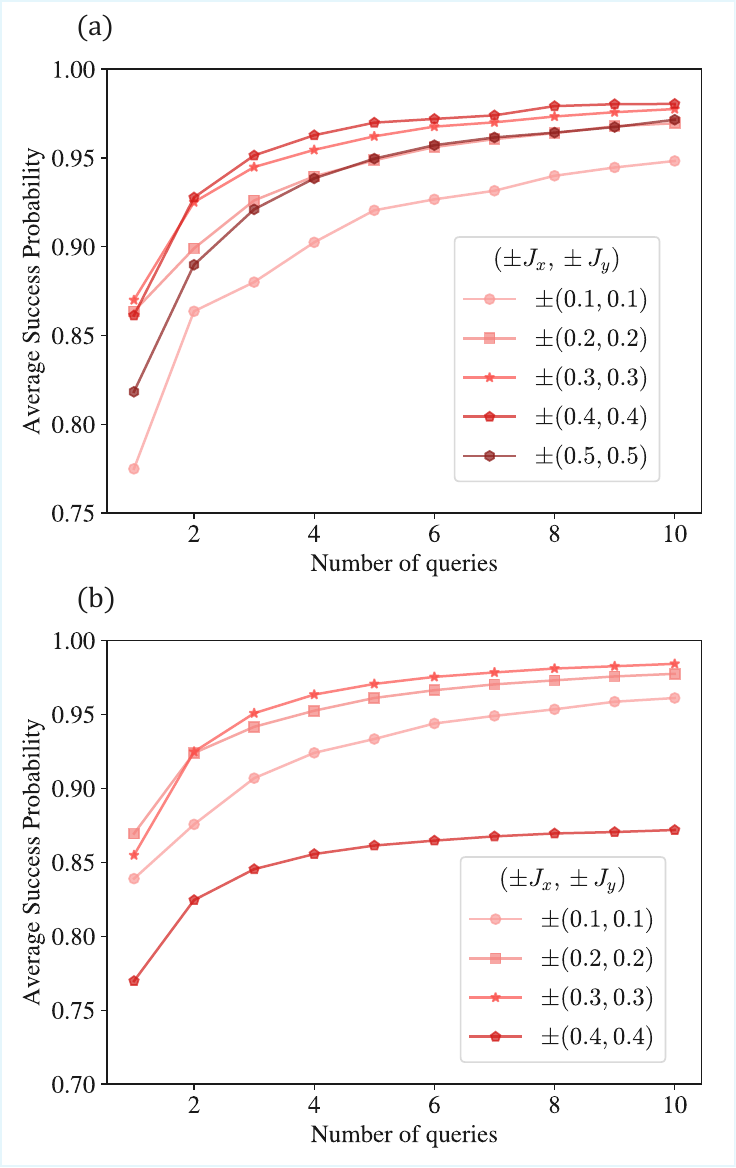}
    \caption{Average success probability for recognizing (a) two-qubit Hamiltonians $J_x X_1 X_2 + J_y Y_1 Y_2$ and (b) three-qubit Hamiltonians $J_x (X_1 X_2+X_2X_3) + J_y (Y_1 Y_2 + Y_2 Y_3)$. Each curve represents a different pair of Hamiltonians $\{\hat{H}_0, \hat{H}_1\}$, indexed by the parameter $(J_x,J_y)$ and $(-J_x, -J_y)$. The $x$-axis represents the number of queries to the unknown unitary, and the $y$-axis represents the optimized average success probability. It shows that the average success probability of Hamiltonian recognition increases with the number of queries, $k$, for all instances of Hamiltonian pairs considered.}
    \label{fig:heisen_2_3qubit}
\end{figure}

\section{Discussion}
In summary, we introduced Hamiltonian recognition, bridging composite hypothesis testing and quantum metrology. For binary recognition of orthogonal single-qubit Hamiltonians, our QSP-based protocol achieves an optimal $O(1/k)$ average error, where $k$ is the number of queries. Remarkably, we showed that two disjoint sets of unitaries can be perfectly discriminated with finite adaptive queries without entanglement. We also developed an optimal quantum algorithm for ternary recognition. 
By identifying the Hamiltonian from a given class, Hamiltonian recognition provides valuable insights into the quantum system's properties and behavior, serving as a foundation for further analysis or as a complementary task to the more comprehensive goal of Hamiltonian learning~\cite{Gebhart2023,Li2024,Haah2024}, quantum property hypothesis testing of quantum dynamics~\cite{Chen2024b}, and quantum device verification~\cite{Eisert_2020}.

While motivated by breaking down a multi-parameter metrology task into the Hamiltonian recognition and single-parameter metrology, it remains challenging and intrinsic to develop protocols for the recognition of multiple non-orthogonal Hamiltonians and incorporating general many-body Hamiltonians~\cite{Holz2015}. Furthermore, while our analysis focuses on the limits of this task in the query model, a complete assessment of its practical feasibility requires further investigation, including investigation on the circuit depth, gate counts, the effects of realistic noise, and hardware connectivity, which are critical for scaling to larger systems. Furthermore, our investigation of distinguishing between two equiprobable multi-qubit Heisenberg-type Hamiltonians indicates the possibility that effective recognition protocols exist for larger quantum systems, which may lead to new quantum algorithms. Our work may also inspire extensions and applications of QSP from both algorithmic and information-processing aspects.

\section*{Acknowledgments}
We acknowledge insightful discussions with Yuxiang Yang, Qiushi Liu, and Zhenhuan Liu, and helpful comments from Benchi Zhao and Guangxi Li. This work was partially supported by the National Key R\&D Program of China (Grant No.~2024YFB4504004), the National Natural Science Foundation of China (Grant No.~12447107), the Guangdong Provincial Quantum Science Strategic Initiative (Grant No.~GDZX2403008, GDZX2303007), the Guangdong Provincial Key Lab of Integrated Communication, Sensing and Computation for Ubiquitous Internet of Things (Grant No.~2023B1212010007), and the Undergraduate Research Project and X Program from HKUST(Guangzhou).

\section*{Data Availability}
The datasets analysed during the current study are available at~\cite{coderepo}.

\section*{Code Availability}
Code for demonstrating binary and tenary Hamiltonian recognition algorithms and related real experiments is available at~\cite{coderepo}. Other code used in this study is available from the corresponding authors upon reasonable request.

\section*{Competing Interests}
The authors declare no competing interests.

\section*{Author Contributions}
C.Z. and S.H. contributed equally to this work. C.Z. and S.H. contributed to the main algorithms of this work. C.Z. and Y.-A.C. contributed to the converse bound. S.H. and L.Z. contributed to the numerical simulation of protocols. X.W. formulated the initial idea and advised the project. All co-authors contributed to the results and the preparation of the manuscript.

%

\appendix

\onecolumngrid

\vspace{1cm}
\begin{center}
{\Large \textbf{Supplementary Information}}
\end{center}

\setcounter{subsection}{0}
\setcounter{table}{0}
\setcounter{figure}{0}

\numberwithin{equation}{section}
\renewcommand{\theproposition}{S\arabic{proposition}}
\renewcommand{\thedefinition}{S\arabic{definition}}
\renewcommand{\thefigure}{S\arabic{figure}}
\setcounter{equation}{0}
\setcounter{table}{0}
\setcounter{section}{0}
\setcounter{proposition}{0}
\setcounter{definition}{0}
\setcounter{figure}{0}

In this Supplementary Information, we present detailed proofs of the theorems in the manuscript ``Optimal Hamiltonian recognition of unknown quantum dynamics''. In Section~\ref{appendix:preliminary}, we review the tester formalism used for quantum channel discrimination and some useful results in quantum signal processing. In Section~\ref{appendix:proof_XZ}, we present detailed proofs of Theorem 1 in the manuscript for binary Hamiltonian recognition. In Section~\ref{appendix:proof_XYZ}, we present detailed proofs of Theorem 3 in the manuscript for ternary Hamiltonian recognition. In Section~\ref{appendix:realdevice_exp}, we give a detailed algorithm for the experiments on a superconducting computer. In Section~\ref{appendix:general_axis}, we present numerical results on the recognition of general non-orthogonal Hamiltonians.

\section{Preliminaries}\label{appendix:preliminary}

\subsection{Tester formalism}
We review the tester formalism for quantum channel discrimination based on the {\Choi} isomorphism~\cite{Bavaresco_2022}. A tester is defined as a set of positive semidefinite operators $T = \{T_0, T_1\},  T_i \in \cL^\dagger (\cH_I \otimes \cH_O)$ where we denote by $\cL^\dagger(\cH)$ the set of Hermitian operators acting on a Hilbert space $\cH$. By taking the trace with $J^{\cN}_{I O}$ and $J^{\cM}_{I O}$, it results in a valid probability distribution according to $p(\cN_{I \to O}|J^{\cN}_{I O}) = \tr(T_0 J^{\cN}_{I O})$ and $p(\cM_{I \to O}|J^{\cM}_{I O}) = \tr(T_1 J^{\cM}_{I O})$. Compared to how POVMs act on quantum states, these testers that act on quantum channels can be regarded as a `measurement' of quantum channels. Supposing two channels $\cN$ and $\cM$ are given with prior probability $p$ and $1-p$, respectively, we can denote the optimal success probability of discriminating $\cN$ and $\cM$ as
\begin{equation}~\label{max_pro_suc_dis_eq}
    P = \max_{\{T_0, T_1\}} p \cdot \tr\left(T_0J^{\cN}_{I O}\right) + (1-p) \cdot \tr\left(T_1J^{\cM}_{I O}\right).
\end{equation}
For simplicity of notations, the trace-and-replace operation in $\cH_X$ is denoted as
\begin{equation}
    _{X}(\cdot) \coloneqq \tr_X (\cdot) \otimes \frac{\idop_X}{d_X},
\end{equation}
where $\idop_X$ denotes the identity operator on $\cH_X$ and $d_X$ is the dimension of $\cH_X$.

\paragraph{Parallel strategy.} This strategy involves sending each component of a multipartite state through independent copies of an unknown channel. In this approach, the outputs from these channels do not interact with each other's inputs. These strategies are characterized by parallel $k$-slot testers $T^{\texttt{PAR}} = \{T_0^{\texttt{PAR}}, T_1^{\texttt{PAR}}\}$, $T_i^{\texttt{PAR}} \in \cL^\dagger(\cH_I \ox \cH_O)$. We set $W^{\texttt{PAR}} \coloneqq \sum^1_{i=0}T_i^{\texttt{PAR}}$ and parallel $k$-slot testers are defined as
\begin{align}
    T_i^{\texttt{PAR}} &\geq 0,~ \forall i,\\
    \tr W^{\texttt{PAR}} &= d_O,\\
    W^{\texttt{PAR}} &= _{O}W^{\texttt{PAR}},
\end{align}
where $O\coloneqq O_1 O_2\cdots O_k$ is a composite of $k$ subsystems.

\paragraph{Sequential strategy.} This strategy is also sometimes called adaptive strategy which involves sending a quantum system through a series of channels in sequence. In this approach, the output from each channel can be used as input for the next one. These strategies are characterized by sequential $k$-slot testers $T^{\texttt{SEQ}} = \{T_0^{\texttt{SEQ}}, T_1^{\texttt{SEQ}}\}$, $T_i^{\texttt{SEQ}} \in \cL^\dagger(\cH_I \ox \cH_O)$. We set $W^{\texttt{SEQ}} \coloneqq  \sum^1_{i=0}T_i^{\texttt{SEQ}}$ and sequential $k$-slot testers are defined as
\begin{align}
    T_i^{\texttt{SEQ}} &\geq 0,~ \forall i,\\
    \tr W^{\texttt{SEQ}} &= d_O,\\
    W^{\texttt{SEQ}} &= _{O_k}W^{\texttt{SEQ}},\\
    _{I_k O_k}W^{\texttt{SEQ}} &= _{O_{k-1} I_k O_k}W^{\texttt{SEQ}},\\
    &\dots \notag \\
    _{I_2 O_2 \dots I_k O_k}W^{\texttt{SEQ}} &= _{O_1 I_2 O_2 \dots I_k O_k}W^{\texttt{SEQ}}.
\end{align}

\paragraph{General strategy.} This strategy involves the most general types of operations that transform multiple quantum channels into a joint probability distribution. In this approach, general strategies do not impose any specific order on the channels, allowing for the most flexible approach. These strategies are characterized by general $k$-slot testers $T^{\texttt{GEN}} = \{T_0^{\texttt{GEN}}, T_1^{\texttt{GEN}}\}$, $T_i^{\texttt{GEN}} \in \cL^\dagger(\cH_I \ox \cH_O)$, which satisfy $p(i|J^{\cN}_{I O},J^{\cM}_{I O}) = \tr\big[T_i^{\texttt{GEN}}(J^{\cN}_{I O} \ox  J^{\cM}_{I O})\big]$. We set $W^{\texttt{GEN}} \coloneqq \sum^1_{i=0} T_0^{\texttt{GEN}}$ and general $k$-slot testers are defined as
\begin{align}
    T_i^{\texttt{GEN}} &\geq 0,~ \forall i,\\
    \tr\big[W^{\texttt{GEN}}(J^{\cN}_{I O} \ox J^{\cM}_{I O})\big] &= 1.
\end{align}

For any strategies mentioned above, we denote $\cS \in \{\texttt{PAR}, \texttt{SEQ}, \texttt{GEN}\}$. By Eq.~\eqref{max_pro_suc_dis_eq}, the optimal success probability of discrimination can be expressed as
\begin{equation}
    P^{\cS} = \max_{\{T_0, T_1\}\in \cT^{^{\cS}}} p \cdot \tr\left[T_0 (J^{\cN}_{I O})^ {\ox k}\right] + (1-p)\cdot \tr\left[T_1 (J^{\cM}_{I O})^ {\ox k}\right].
\end{equation}

\subsection{Quantum signal processing}
Quantum signal processing (QSP) is a technique that extends the concepts of composite pulse sequences proposed by Low, Yoder, and Chuang~\cite{low_optimal_2016,Low_2016}. A common QSP structure consists of four key components: (i) the signal unitary $W$, (ii) the QSP phase sequence $\vec{\phi}\coloneqq (\phi_0, \phi_1,\cdots, \phi_k)$, (iii) the signal processing operators $S(\phi_j)$, and (iv) the measurement basis $M$ in which the desired polynomial is obtained. A QSP circuit is constructed by interleaving the signal unitary $W$ with the signal processing operators $S(\phi_j)$, followed by a projective measurement in the basis $M$. The key feature of QSP is that the signal rotation $W$ always rotates through a fixed angle $\theta$, while the angles for the signal processing rotations $S$ can be chosen to fit desired functions of $\theta$. For a comprehensive review of QSP, we refer the reader to Ref.~\cite{Martyn_2021}. A widely studied and employed signal unitary is
\begin{equation}
W_x(a) \coloneqq \left(\begin{array}{cc}
a & i \sqrt{1-a^2} \\
i \sqrt{1-a^2} & a
\end{array}\right) = e^{i \arccos (a) X},
\end{equation}
which is a rotation around $x$-axis by angle $\theta = 2\arccos(a)$. We will also use $W_x(\theta)$ to denote $W_x(a)$. 
The function-fitting capability of a simple QSP structure is characterized by the following theorem.
\begin{theorem}[\cite{gilyen_quantum_2018}, Theorem 4]
There exists a QSP phase sequence $\vec{\phi} \in \mathbb{R}^{d+1}$ such that
\begin{equation}
e^{i Z\phi_0} \prod_{k=1}^d W(a) e^{iZ\phi_k} = \left(
\begin{array}{cc}
P(a) & i Q(a) \sqrt{1-a^2} \\
i Q^*(a) \sqrt{1-a^2} & P^*(a)
\end{array}\right)
\end{equation}
for $a \in[-1,1]$, and a $\vec{\phi}$ exists for any polynomials $P, Q$ in a such that:
\begin{itemize}
    \item[(i)] $\mathrm{deg}(P) \leq k,~\mathrm{deg}(Q) \leq k-1$,
    \item[(ii)] $P$ has parity $k$ mod $2$ and $Q$ has parity $(d-1)$ mod $2$,
    \item[(iii)] $|P|^2 + (1 - a^2)|Q|^2 = 1$.
\end{itemize}
\end{theorem}

Another useful result is
\begin{theorem}[\cite{Martyn_2021}, Theorem 9]~\label{lem:ZX_qsp}
There exists a QSP phase sequence $\vec{\phi} \in \mathbb{R}^{d+1}$ such that
\begin{equation}
    {\rm Poly}(a) = \bra{0} e^{iZ\phi_0} \prod_{j=1}^k W_x(a) e^{iZ\phi_j}\ket{0}
\end{equation}
for $a \in[-1,1]$, and for any polynomial $\mathrm{Poly} \in \mathbb{C}[a]$ if and only if the following conditions hold:
\begin{itemize}
    \item[(i)] $\mathrm{deg}(\mathrm{Poly}) \leq k$,
    \item[(ii)] $\mathrm{Poly}$ has parity $k \bmod 2$,
    \item[(iii)] $\forall a \in[-1,1],|\mathrm{Poly}(a)| \leq 1$,
    \item[(iv)] $\forall a \in(-\infty,-1] \cup[1, \infty),|\mathrm{Poly}(a)| \geq 1$,
    \item[(v)] if $k$ is even, then $\forall a \in \mathbb{R}, \mathrm{Poly}(i a) \mathrm{Poly}^*(i a) \geq 1$.
\end{itemize}
\end{theorem}

\section{Proof of Theorem 1}\label{appendix:proof_XZ}
In order to prove Theorem 1, we first prove the following lemmas.

\begin{lemma}~\label{lem:err_poly}
For any odd $k$, there exists a QSP phase sequence $\vec{\phi} \in \mathbb{R}^{k+1}$ such that
\begin{equation}
    \left|\bra{0} e^{iZ\phi_0} \prod_{j=1}^k W_x(\theta) e^{iZ\phi_j}\ket{0} \right|^2 = \sum_{l=-k}^{k} \frac{k-|l|+1}{(k+1)^2} e^{il\theta}.
\end{equation}
\end{lemma}
\begin{proof}
Denoting $f_k(\theta) = \sum_{l=-k}^{k} \frac{k-|l|+1}{(k+1)^2} e^{il\theta}$, we first note that the following polynomial $P_k(\theta)$
\begin{equation}
    P_k(\theta) \coloneqq \frac{1}{k+1} \sum_{l=-k,\,l\text{ odd}}^{k} e^{il\theta/2} 
\end{equation}
satisfies $P_k(\theta)P^*_k(\theta) = f_k(\theta)$ because
\begin{equation}
    P_k(\theta)P^*_k(\theta) = \frac{1}{(k+1)^2}\sum_{\substack{l=-k,\\l\text{ odd}}}^k\sum_{\substack{m=-k,\\m\text{ odd}}}^k e^{i(l+m)\theta/2} = \frac{1}{(k+1)^2}\sum_{j=-k}^k \sum_{2j=l+m} e^{i j \theta}.
\end{equation}
By checking how many pairs $(l,m)$ such that $2j = l+m$ where $-k \leq j\leq k$, we have that $P_k(\theta)P^*_k(\theta) = P_k(\theta)$. It is easy to see that $P_k(\theta)\in \mathbb{C}[e^{-i\theta/2}, e^{i\theta/2}]$ and has parity $1$ because $k$ is odd. 
Further, note that
\begin{equation}
    P_k(\theta) = \frac{2}{k+1} \sum_{l=1,\,l\text{ odd}}^{k} \cos(l\theta/2)
    = \frac{2}{k+1} \sum_{l=1,\,l\text{ odd}}^{k} T_l\left(\cos(\theta / 2)\right),
\end{equation}
where $T_l$ is the Chebyshev polynomial of the first kind with degree $l$. Denoting $a \coloneqq \cos (\theta/2)$, $P_k$ can be rewritten as 
\begin{equation}~\label{Eq:zxz poly}
    P_k(a) = \frac{2}{k+1} \sum_{l=1,\,l\text{ odd}}^{k} T_l(a).
\end{equation}
Such reconstruction gives 
\begin{enumerate}
    \item [-] $P_k(a) \in \mathbb{C}[a]$ with degree $k$,
    \item [-] $P_k$ has parity $1$, since $T_l$ has parity $l$ mod $2$,
    \item [-] For all $a \in [-1, 1]$, $|P_k(a)| \leq \frac{2}{k+1} \sum_l |T_l(a)| \leq 1$,
    \item [-] For all $a \in [1, +\infty)$, $|P_k(a)| \geq 1$ by Lemma~\ref{lem:odd poly >= 1}.
\end{enumerate}
Therefore, $P_k$ satisfies all conditions in Theorem~\ref{lem:ZX_qsp}, and hence there exists a QSP phase sequence $\vec{\phi}$ such that 
\begin{equation}
     P_k(a) = \bra{0} e^{i Z\phi_0} \prod_{j=1}^k W_x(a) e^{i Z \phi_j}\ket{0} 
     \implies f_k(\theta) = \left|\bra{0} e^{i Z\phi_0} \prod_{j=1}^k W_x(\theta) e^{i Z\phi_j}\ket{0} \right|^2,
\end{equation}
which completes the proof.
\end{proof}

\begin{lemma}~\label{lem:odd poly >= 1}
    For $a \in \RR$ such that $|a| \geq 1$ and an odd $k$, $P_k$ in Eq.~\eqref{Eq:zxz poly} satisfies $|P_k(a)| \geq 1$.
\end{lemma}
\begin{proof}
Since $P_k$ is an odd function, it suffices to prove the positive case. In the rest of the proof, we assume $a \geq 1$. By definition of the Chebyshev polynomials, in such a case $T_l(a) = \cosh\left( l \cosh^{-1}(a) \right)$. One can denote $\phi = \cosh(a)$ such that $T_l(a) = \cosh(l\phi)$ and hence 
\begin{align}
    P_k(a) &= \frac{2}{k + 1} \sum_{l=1,\,l\text{ odd}}^{k} \cosh(l \phi)
    = \frac{2}{k + 1} \sum_{j=0}^{(k - 1) / 2} \cosh\left((2j + 1) \phi\right) \\
    &= \frac{2}{k + 1} \frac{\sinh \left((k + 1) \phi\right)}{2 \sinh (\phi)} = \frac{\sinh  \left((k + 1) \phi\right)}{(k + 1) \sinh (\phi)}.
\end{align}
To show that the RHS of the equation is no less than $1$, we further define $f_\phi(k) = \sinh  \left((k + 1) \phi\right) / \left((k + 1) \sinh (\phi)\right)$. Since $f_\phi'(k) \geq 0$ and $f(0) = 1$, we conclude that $f(k) \geq 1$ for $k \geq 0$, and the statement follows.
\end{proof}

\begin{lemma}\label{lem:Ck_geq_id}
Let $k \in \mathbb{N}^+$, and
\begin{equation}
    \Omega_H^{(k)}(\pi) = \frac{1}{2\pi} \int_{0}^{2\pi} \Big[\left(e^{-iH\theta/2}\ox \idop\right)\dketbra{\idop}{\idop} \left(e^{iH\theta/2}\ox \idop\right)\Big]^{\ox k} \mathrm{d}\theta,
\end{equation}
where $H\in\{X,Y,Z\}$. It satisfies
\begin{equation}
    \Omega_H^{(k)}(\pi) - \frac{1}{k+1}\dketbra{\idop}{\idop}^{\ox k} \geq 0.
\end{equation}
\end{lemma}
\begin{proof}
Consider
\begin{equation}
    \left(e^{-iZ\theta/2}\ox \idop\right)\dketbra{\idop}{\idop} \left(e^{iZ\theta/2}\ox \idop\right) = \sum_{j,l=0,1} e^{i\theta(j-l)}\ketbra{jj}{ll}.
\end{equation}
Taking the $k$-fold tensor product, we get
\begin{equation}
    \Big[\left(e^{-iZ\theta/2}\ox \idop\right)\dketbra{\idop}{\idop} \left(e^{iZ\theta/2}\ox \idop\right) \Big]^{\ox k} = \sum_{j_1,l_1,...,j_k,l_k=0,1} e^{i\theta\sum_{\ell=1}^k(j_\ell - l_\ell)} \ketbra{j_1j_1\cdots j_kj_k}{l_1l_1\cdots l_kl_k}.
\end{equation}
Taking integral over $\theta$ from $0$ to $2\pi$, we see that
\begin{equation}
\frac{1}{2\pi}\int_0^{2\pi} e^{i m\theta} d\theta = \delta_{m,0}.
\end{equation}
Thus, only terms with $\sum_{\ell=1}^k(j_\ell-l_\ell) = 0$ survive and we can rewrite
\begin{equation}
    \Omega_{Z}^{(k)}(\pi) = \sum_{\lambda\in \text{irrep}} \sum_{\mathbf{j}_{\lambda} \in \cB_k(\lambda)} \sum_{\mathbf{k}_{\lambda} \in \cB_k(\lambda)} \ketbra{\mathbf{j}_{\lambda}\mathbf{j}_{\lambda}}{\mathbf{k}_{\lambda}\mathbf{k}_{\lambda}},
\end{equation} 
where $\lambda$ runs over the irreducible representations of $\mathrm{U}(1)$: $\lambda(e^{i\theta}) = e^{i \lambda \theta}$ and 
\begin{equation}
    \cB_k(\lambda)\coloneqq \left\{~\mathbf{j}_{\lambda}~\left|\right.~\mathbf{j}_{\lambda}\in \{0,1\}^k,~R_z(\theta)^{\ox k}\ket{\mathbf{j}_{\lambda}} = e^{i \lambda \theta}\ket{\mathbf{j}_{\lambda}} \right\}.
\end{equation}
Here we consider the group action $\mathrm{U}(1) \rightarrow \mathrm{GL}((\mathbb{C}^2)^{\ox 2k}),~ e^{i\theta} \mapsto R_z(\theta)^{\ox k} \ox \idop^{\ox k}$. Notably, by direct calculation, there are $k+1$ irreducible representations for $k$-fold tensor product. Meanwhile, we can also express $\dketbra{\idop}{\idop}^{\ox k}$ as 
\begin{equation}
    \dketbra{\idop}{\idop}^{\ox k} =\sum_{\lambda\in \text{irrep}} \sum_{\lambda' \in \text{irrep}} \sum_{\mathbf{j}_{\lambda} \in \cB_k(\lambda)} \sum_{\mathbf{k}_{\lambda'} \in \cB_k(\lambda')} \ketbra{\mathbf{j}_{\lambda}\mathbf{j}_{\lambda}}{\mathbf{k}_{\lambda'}\mathbf{k}_{\lambda'}}.
\end{equation}
Consider $\max\{t~|~\Omega_{Z}^{(k)}(\pi) - t \cdot \dketbra{\idop}{\idop}^{\ox k} \geq 0\}$. We have
\begin{equation}
    t = \max_{\ket{v}} \frac{\bra{v}\Omega_{Z}^{(k)}(\pi)\ket{v}}{\bra{v} (\dketbra{\idop}{\idop}^{\ox k})\ket{v}}.
\end{equation}
Choose a pure state
\begin{equation}
    \ket{v} = \sum_{\lambda\in \text{irrep}} p_{\lambda} \sum_{\mathbf{j}_{\lambda}\in \cB_k(\lambda)} \ket{\mathbf{j}_{\lambda}\mathbf{j}_{\lambda}},
\end{equation}
where $\sum_{\lambda} |p_{\lambda}|^2 = 1$. It follows that
\begin{equation}
\begin{aligned}
    \bra{v}\Omega_{Z}^{(k)}(\pi)\ket{v} &= \sum_{\lambda\in \text{irrep}} p^*_{\lambda} d_{\lambda} \sum_{\mathbf{k}_{\lambda}\in \cB_k(\lambda)} \bra{\mathbf{k}_{\lambda}\mathbf{k}_{\lambda}} \left(\sum_{\lambda\in \text{irrep}} p_{\lambda} \sum_{\mathbf{j}_{\lambda}\in \cB_k(\lambda)} \ket{\mathbf{j}_{\lambda}\mathbf{j}_{\lambda}}\right)\\
    &= \sum_{\lambda\in \text{irrep}} p^*_{\lambda} p_{\lambda} d_{\lambda}^2,
\end{aligned}
\end{equation}
where $d_{\lambda} = |\cB_k(\lambda)|$ and
\begin{equation}
\begin{aligned}
    \bra{v} (\dketbra{\idop}{\idop}^{\ox k})\ket{v} &= \sum_{\lambda\in \text{irrep}} p^*_{\lambda} d_{\lambda} \sum_{\lambda'\in \text{irrep}} \sum_{\mathbf{k}_{\lambda'}\in \cB_k(\lambda')} \bra{\mathbf{k}_{\lambda'}\mathbf{k}_{\lambda'}}  \left(\sum_{\lambda\in \text{irrep}} p_{\lambda} \sum_{\mathbf{j}_{\lambda}\in \cB_k(\lambda)} \ket{\mathbf{j}_{\lambda}\mathbf{j}_{\lambda}}\right)\\
    &= \sum_{\lambda\in \text{irrep}} p^*_{\lambda} d_{\lambda} \sum_{\lambda\in \text{irrep}} p_{\lambda} d_{\lambda}.
\end{aligned}
\end{equation}
Since $\ket{v}$ is feasible pure state, it follows that
\begin{equation}
    t \geq \frac{\sum_{\lambda}p^*_{\lambda} p_{\lambda} d_{\lambda}^2}{|\sum_{\lambda} p_{\lambda} d_{\lambda}|^2} \geq \frac{\sum_{\lambda} |p_{\lambda}|^2 d_{\lambda}^2}{\sum_{\lambda}1 \cdot \sum_{\lambda} |p_{\lambda} d_{\lambda}|^2} = \frac{1}{k+1},
\end{equation}
where we used the Cauchy inequality for the second inequality. Therefore, we have $t\geq \frac{1}{k+1}$ and thus, 
\begin{equation}\label{Eq:CkZ_id_ineq}
    \Omega_{Z}^{(k)}(\pi) - \frac{1}{k+1} \dketbra{\idop}{\idop}^{\ox k} \geq 0.
\end{equation}

Followed by Eq.~\eqref{Eq:CkZ_id_ineq}, we have for any unitary operator $U$,
\begin{equation}
    (U\ox \overline{U})^{\ox k} \Omega_{Z}^{(k)}(\pi)(U^{\dag} \ox U^T)^{\ox k} - \frac{1}{k+1} \dketbra{\idop}{\idop}^{\ox k} \geq 0,
\end{equation}
where we used the fact that $\dketbra{\idop}{\idop}$ is invariant under $U\ox \overline{U}$. Notice that
\begin{equation}
\begin{aligned}
    &(U\ox \overline{U}) \left(e^{-iZ\theta/2}\ox \idop\right)\dketbra{\idop}{\idop} \left(e^{iZ\theta/2}\ox \idop\right)(U^{\dag} \ox U^T)\\
    =&\; (U\ox \overline{U}) \left(e^{-iZ\theta/2}\ox \idop\right) (U^{\dag} \ox U^T) \dketbra{\idop}{\idop} (U\ox \overline{U}) \left(e^{iZ\theta/2}\ox \idop\right) (U^{\dag} \ox U^T)\\
    =&\; \left(U e^{-iZ\theta/2}U^{\dag} \ox \idop \right) \dketbra{\idop}{\idop} \left(U e^{iZ\theta/2} U^{\dag}\ox \idop \right).
\end{aligned}
\end{equation}
By choosing specific $U$, we can conclude that $\Omega_{X}^{(k)}(\pi) - \frac{1}{k+1} \dketbra{\idop}{\idop}^{\ox k} \geq 0$ and $\Omega_{Y}^{(k)}(\pi) - \frac{1}{k+1} \dketbra{\idop}{\idop}^{\ox k} \geq 0$ as well.
\end{proof}

\begin{lemma}~\label{lem:even_poly}
For any even number $k\geq2$, there exist real Laurant polynomials $F,G\in \mathbb{R}[\omega^{-1}, \omega]$ where $\omega = e^{i\theta/2}$ such that 
\begin{enumerate}
    \item $F(1)=G(1)=\frac{1}{\sqrt{2}}$,
    \item $F(\omega)F(\omega^{-1}) + G(\omega)G(\omega^{-1}) = 1$,
    \item $F(\omega)G(\omega^{-1}) + F(\omega^{-1})G(\omega) = \frac{1}{k+1} \sum_{l=-k,\,l\text{ even}}^{k} \omega^{l}$.
\end{enumerate}
\end{lemma}
\begin{proof}
Denote $P(\omega) \coloneqq F(\omega) + G(\omega)$ and $Q(\omega) \coloneqq F(\omega) - G(\omega)$. It follows that the above conditions can be equivalently written as
\begin{enumerate}
    \item $P(1) = \sqrt{2},~Q(1) = 0$,
    \item $P(\omega)P(\omega^{-1}) = 1 + \frac{1}{k+1}\sum_{l=-k,\,l\text{ even}}^{k} \omega^{l}$,
    \item $Q(\omega)Q(\omega^{-1}) = 1 - \frac{1}{k+1}\sum_{l=-k,\,l\text{ even}}^{k} \omega^{l}$.
\end{enumerate}
Denote
\begin{equation}
\begin{aligned}
&H(z)\coloneqq1+ \frac{1}{k+1}\sum_{l=-k,\,l\text{ even}}^{k} z^l
=\frac{1}{(k+1)z^k}\prod_{H(\zeta)=0}(z-\zeta),\\
&L(z)\coloneqq1- \frac{1}{k+1}\sum_{l=-k,\,l\text{ even}}^{k} z^l
=-\frac{1}{(k+1)z^k}\prod_{L(\zeta)=0}(z-\zeta).
\end{aligned}
\end{equation}
In the following, we will prove the lemma by constructing Laurant polynomials $P,Q\in \mathbb{R}[\omega^{-1},\omega]$ that satisfy the above three conditions.

First, notice that $\forall z\in\mathbb{C},\ H(z)=H(z^{-1}),~H(\zeta)=0\implies H(\overline\zeta)=0$, and $|z|=1\text{ or } z\in\mathbb R\implies H(z)\ne0$. We have
\begin{equation}
H(z)=\frac{1}{(k+1)z^{k}} \prod_{H(\zeta)=0,|\zeta|>1,\operatorname{Im}\zeta>0} (z-\zeta)(z-\overline\zeta)(z-\zeta^{-1})(z-\overline\zeta^{-1}).
\end{equation}
Let 
\begin{align}P(z)=&\; a\prod_{H(\zeta)=0,|\zeta|>1,\operatorname{Im}\zeta>0}(z-\zeta)(z-\overline\zeta)\\
=&\; a\prod_{H(\zeta)=0,|\zeta|>1,\operatorname{Im}\zeta>0}(z^2-(\zeta+\overline\zeta)z+\zeta\overline\zeta)\in\mathbb R[z,a].
\end{align}
It follows that
\begin{align}
P(z)P(z^{-1})
=&\; a^2\prod_{H(\zeta)=0,|\zeta|>1,\operatorname{Im}\zeta>0}
(z-\zeta)(z-\overline\zeta)(z^{-1}-\zeta)(z^{-1}-\overline\zeta)\\
=&\; {a^2}\prod_{H(\zeta)=0,|\zeta|>1,\operatorname{Im}\zeta>0}
(z-\zeta)(z-\overline\zeta)\frac{\zeta\overline\zeta}{z^2}(\zeta^{-1}-z)(\overline\zeta^{-1}-z)\\
=&\; \frac{a^2}{z^k}\prod_{H(\zeta)=0,|\zeta|>1,\operatorname{Im}\zeta>0}
{\zeta\overline\zeta}(z-\zeta)(z-\overline\zeta)(\zeta^{-1}-z)(\overline\zeta^{-1}-z)\\
=&\; {a^2(k+1)H(z)}\prod_{H(\zeta)=0,|\zeta|>1,\operatorname{Im}\zeta>0}
{\zeta\overline\zeta}.
\end{align}
Then by choosing
\begin{equation}
    a=\Big[(k+1)\prod_{H(\zeta)=0,|\zeta|>1,\operatorname{Im}\zeta>0}\zeta\overline\zeta\Big]^{-1/2}\in\mathbb R,
\end{equation}
we have
\begin{equation}
    P(z)P(z^{-1})=H(z)\text{ and }P(z)\in\mathbb R[z],
\end{equation}
$H(1)=2$ and $P(z) = \sqrt{2}$. Second, by calculating
\begin{equation}
\frac{d}{dz}L(z)|_{z=\pm1} = \frac{1}{k+1} \frac{d}{dz} \sum_{l=-k,\,l\text{ even}}^{k} z^l \Big|_{z=\pm1} = \frac{1}{k+1}\sum_{l=-k,\,l\text{ even}}^{k} l z^{l-1} \Big|_{z=\pm1} = \frac{1}{k+1}\sum_{l=-k,\,l\text{ even}}^{k} l z \Big|_{z=\pm1} = 0,
\end{equation}
and
\begin{equation}
\frac{d^2}{dz^2}L(z)|_{z=\pm1} = \frac{1}{k+1} \frac{d^2}{dz^2} \sum_{l=-k,\,l\text{ even}}^{k} z^l|_{z=\pm1} = \frac{1}{k+1}\sum_{l=-k,\,l\text{ even}}^{k} l(l-1) z^{l-2} \Big|_{z=\pm1} = \frac{1}{k+1}\sum_{l=-k,\,l\text{ even}}^{k} l(l-1) > 0,
\end{equation}
we know the zero multiplicity of the polynomial $L(z)$ at $z=\pm 1$ is $2$. Noticing that $\forall z\in\mathbb{C},\ L(z)=L(z^{-1}),~L(\zeta)=0\implies L(\overline\zeta)=0$, and $(|z|=1\text{ or } z\in\mathbb R)\implies (L(z)\ne0\text{ or }z=\pm1)$, we can rewrite
\begin{equation}
L(z) = -\frac{(z^2-1)^2}{(k+1)z^{k}} \prod_{L(\zeta)=0,|\zeta|>1,\operatorname{Im}\zeta>0} (z-\zeta)(z-\overline\zeta)(z-\zeta^{-1})(z-\overline\zeta^{-1}).
\end{equation}
Let
\begin{equation}
\begin{aligned}
Q(z)=&\; b(z^2-1)
\prod_{L(\zeta)=0,|\zeta|>1,\operatorname{Im}\zeta>0}
(z-\zeta)(z-\overline\zeta)\\
=&\; b(z^2-1)
\prod_{L(\zeta)=0,|\zeta|>1,\operatorname{Im}\zeta>0}
(z^2-(\zeta+\overline\zeta)z+\zeta\overline\zeta)\in\mathbb R[z,b].
\end{aligned}
\end{equation}
It follows that
\begin{align}
Q(z)Q(z^{-1})
=&\; b^2(z^2-1)(z^{-2}-1)\prod_{L(\zeta)=0,|\zeta|>1,\operatorname{Im}\zeta>0}
(z-\zeta)(z-\overline\zeta)(z^{-1}-\zeta)(z^{-1}-\overline\zeta)\\
=&\; -\frac{b^2(z^2-1)^2}{z^2}\prod_{L(\zeta)=0,|\zeta|>1,\operatorname{Im}\zeta>0}
(z-\zeta)(z-\overline\zeta)\frac{\zeta\overline\zeta}{z^2}(\zeta^{-1}-z)(\overline\zeta^{-1}-z)\\
=&\; -\frac{b^2(z^2-1)^2}{z^k}\prod_{L(\zeta)=0,|\zeta|>1,\operatorname{Im}\zeta>0}
{\zeta\overline\zeta}(z-\zeta)(z-\overline\zeta)(\zeta^{-1}-z)(\overline\zeta^{-1}-z)\\
=&\; {b^2(k+1)L(z)}\prod_{L(\zeta)=0,|\zeta|>1,\operatorname{Im}\zeta>0}
{\zeta\overline\zeta}.
\end{align}
By choosing 
\begin{equation}
b = \pm\Big[(k+1)\prod_{L(\zeta)=0,|\zeta|>1,\operatorname{Im}\zeta>0}\zeta\overline\zeta\Big]^{-1/2}\in\mathbb R,  
\end{equation}
we have
\begin{equation}
    Q(z)Q(z^{-1})=L(z)\text{ and }Q(z)\in\mathbb R[z],
\end{equation}
and $Q(1) = 0$. Hence, we complete the proof.
\end{proof}

Now, we are ready to prove Theorem \textcolor{blue}{1}.

\renewcommand\theproposition{\textcolor{blue}{1}}
\begin{theorem}[Binary Hamiltonian recognition]\label{thm:recog_XZ}
For $H_0 = \vec{n}_0\cdot \vec{\sigma}$ and $H_1 = \vec{n}_1\cdot \vec{\sigma}$ with prior probabilities $p$ and $1-p$, respectively, if $\vec{n}_0 \perp \vec{n}_1$, the optimal average success probability of recognizing $\{H_0,H_1\}$ using $k$ queries is given by
\begin{equation*}
{\rm Suc}_k^{\texttt{GEN}}\left(\{H_0, H_1\},  \{ p, 1-p\}, \pi \right) = \frac{k+\max \{p,1-p\}}{k+1}.
\end{equation*}
This limit is globally optimal, holding for all $k$-query strategies, and is saturated by a sequential protocol without entanglement.
\end{theorem}
\renewcommand{\theproposition}{S\arabic{proposition}}

We will first prove the result for $\vec{n}_0 = [1,0,0]$ and $\vec{n}_1 = [0,0,1]$ and reduce the general cases to this one. For $H_0 = X$ and $H_1 = Z$, we are going to show that ${\rm Suc}_k^{\texttt{SEQ}}\left(\{H_0, H_1\},  \{ p, 1-p\}, \pi \right) \geq (k+\max\{p,1-p\})/(k+1)$ by utilizing the framework of quantum signal processing to construct a protocol for the recognition explicitly; and then prove ${\rm Suc}_k^{\texttt{SEQ}}\left(\{H_0, H_1\},  \{ p, 1-p\}, \pi \right) \leq (k+\max\{p,1-p\})/(k+1)$ by using the SDP dual formalism of the problem. In the following, the proof is presented for the range $\frac{1}{2} \le p < 1$, with the rationale to be detailed later.

\subsection{Achievability}

\paragraph{When $k$ is odd.} We construct a protocol for recognizing $\{X,Z\}$ as follows. First, we explicitly chose $k+1$ angles $\phi_j,j=0,1,...,k$. Given an unknown unitary $e^{-iH\theta}$, we sequentially apply $e^{iZ\phi_k} e^{iH\theta} e^{iZ\phi_{k-1}} e^{iH\theta}\cdots e^{iZ\phi_1} e^{iH\theta} e^{iZ\phi_0}$ to the zero state. Then we measure the state on a computational basis. If outcome is $0$, we decide $H$ is $Z$; if outcome is $1$, we decide $H$ is $X$. We will explain later how we chose $\vec{\phi}$. It can be seen that the composed unitary is reminiscent of the QSP framework, where $U_H$ is the signal unitary, and one has control over the angles $\vec{\phi}$, which are called QSP phase sequences. Therefore, we denote the composed unitary as
\begin{equation}~\label{Eq:qsp_UH_theta}
    {\rm QSP}_{H,k}(\theta) \coloneqq R_z(\phi_0) \prod_{j=1}^k U_{H}(\theta) R_z(\phi_j).
\end{equation}
If the unknown unitary is $U_Z(\theta) = e^{-i Z \theta}$, we have
\begin{equation}
    {\rm QSP}_{Z,k}(\theta/2) = R_z(\phi_0)\prod_{j=1}^k R_z(\theta)R_z(\phi_j) = R_z(n\theta + \phi_{:k}),
\end{equation}
where $\phi_{:k} \coloneqq \sum_{j=1}^k \phi_j$. It follows that for any chosen parameters $\vec{\phi}$, it always holds that
\begin{equation}\label{Eq:odderr_Z}
    \big|\bra{1} {\rm QSP}_{Z,k}(\theta) \ket{0}\big| = 0\; \text{ and }\; \big|\bra{0} {\rm QSP}_{Z,k}(\theta) \ket{0}\big| = 1, \; \forall \theta \in[0, \pi].
\end{equation}
If the unknown unitary is $U_X(\theta) = e^{-i X \theta}$, we have
\begin{equation}
    {\rm QSP}_{X,k}(\theta/2) = R_z(\phi_0)\prod_{j=1}^k R_x(\theta) R_z(\phi_j).
\end{equation}
According to Lemma~\ref{lem:err_poly}, we can choose a specific $\vec{\phi}$ for our protocol such that
\begin{equation}\label{Eq:odderr_X}
    \left|\bra{0} {\rm QSP}_{X,k}(\theta/2) \ket{0}\right|^2 = \sum_{l=-k}^{k} \frac{k-|l|+1}{(k+1)^2} e^{il\theta}.
\end{equation}
Combining Eq.~\eqref{Eq:odderr_X} and Eq.~\eqref{Eq:odderr_Z}, the average success probability is lower bounded by

\begin{equation}\label{Eq:XZ_odd_suc_geq}
\begin{aligned}
    {\rm Suc}_k^{\texttt{SEQ}}\left(\{X, Z\}, \{p, 1-p\},\pi \right) &\geq \frac{1}{4\pi^2} \int_{0}^{2\pi} \int_{0}^{2\pi} p\left|\bra{0} {\rm QSP}_{Z,k}(\theta_0/2) \ket{0}\right|^2 + (1-p)\left|\bra{1} {\rm QSP}_{X,k}(\theta_1/2) \ket{0}\right|^2 \mathrm{d}\theta_0 \mathrm{d}\theta_1 \\
     &= \frac{1}{2\pi} \int^{2\pi}_0 p + (1-p)\Big[\frac{k}{k+1} - \frac{2}{(k+1)^2}\sum_{l=1}^{k}(k-l+1)\cos(l\theta_1) \mathrm{d}\theta_1\Big]\\
     &= \frac{k+p}{k+1} -  \frac{1}{2\pi} \frac{2(1-p)}{(k+1)^2}\sum_{l=1}^{k}(k-l+1) \int^{2\pi}_0 \cos(l\theta_1) \mathrm{d}\theta_1\\
     &= \frac{k+p}{k+1},
\end{aligned}
\end{equation}
where in the second line we used the fact that $\left(e^{il\theta} + e^{-il\theta}\right) = 2\cos(l\theta)$ and
\begin{equation}
    1-\sum_{l=-k}^{k} \frac{k-|l|+1}{(k+1)^2} e^{il\theta} = \frac{k}{k+1} - \frac{2}{(k+1)^2}\sum_{l=1}^{k}(k-l+1)\cos(l\theta).
\end{equation}
\begin{remark}
For $ \frac{1}{2} \le p < 1$, the average success probability is given by $(k+p)/(k+1)$. In the case where $0 < p < \frac{1}{2}$, the recognition problem exhibits a symmetry under the interchange of $X$ and $Z$ via the Hadamard transformation $H$, since $HXH = Z$ and $HZH = X$. Consequently, the scenario is equivalent to one with parameter $1-p$, and the average success probability becomes $(k+1-p)/(k+1)$. Subsequent proofs follow from an analogous reasoning.
\end{remark}

\paragraph{When $k$ is even.} 
Let $k=2n$, and we construct a similar protocol with a quantum circuit read as
\begin{equation}
    {\rm QSP}_{H,k}(\theta/2) \coloneqq R_x(\phi_0) \left( \prod_{j=1}^n R_H(\theta) R_x(\phi_j)\right) R_z(\pi)R_x(-\phi_{n}) \left( \prod_{j=n-1}^{0} R_H(\theta) R_x(-\phi_j) \right),
\end{equation}
where we set $\sum_{j=0}^n \phi_j = -\pi/2$. Let the input state be the zero state, and measure the output state on a computational basis. If outcome is $0$, we decide $H$ is $Z$; if outcome is $1$, we decide $H$ is $X$.

If the unknown unitary is $U_Z(\theta) = e^{-i Z \theta}$, we have
\begin{equation}~\label{V_ev_zk_theta}
    {\rm QSP}_{Z,k}(\theta/2) = R_x(\phi_0) \left( \prod_{j=1}^n R_z(\theta) R_x(\phi_j)\right) R_z(\pi)R_x(-\phi_{n}) \left( \prod_{j=n-1}^{0} R_z(\theta) R_x(-\phi_j) \right).
\end{equation}
Noticing the relationship between different types of QSP, we have
\begin{equation}
    R_x(\phi_0) \left( \prod_{j=1}^n R_z(\theta) R_x(\phi_j) \right) = H R_z(\phi_0) \left( \prod_{j=1}^n R_x(\theta) R_x(\phi_j) \right) H = \left(\begin{array}{cc}
F(\omega) & iG(\omega) \\
iG(\omega^{-1}) & F(\omega^{-1})
\end{array}\right),
\end{equation}
where $F,G$ are real Laurent polynomials $F,G \in R[\omega, \omega^{-1}]$ with parity $n$ mod $2$, and $\omega=e^{i\theta/2}$~\cite{Martyn_2021}. Notice that
\begin{equation}
\begin{aligned}
    R_x(-\phi_n) \left( \prod_{j=n-1}^0 R_z(\theta) R_x(-\phi_j) \right) &= X R_x(-\phi_n)\left( \prod_{j=n-1}^0 R_z(-\theta) R_x(-\phi_j) \right)X\\
    &= X\left(\begin{array}{cc}
    F(\omega^{-1}) & -iG(\omega) \\
    -iG(\omega^{-1}) & F(\omega)
    \end{array}\right)X\\
    &= \left(\begin{array}{cc}
    F(\omega) & -iG(\omega^{-1}) \\
    -iG(\omega) & F(\omega^{-1})
    \end{array}\right).
\end{aligned}
\end{equation}
It follows that
\begin{equation}
\begin{aligned}
{\rm QSP}_{Z,k}(\theta/2) &= \left(\begin{array}{cc}
    F(\omega) & iG(\omega) \\
    iG(\omega^{-1}) & F(\omega^{-1})
    \end{array}\right)
    \left(\begin{array}{cc}
    -i &  0 \\
    0 & i
    \end{array}\right)
    \left(\begin{array}{cc}
    F(\omega) & -iG(\omega^{-1}) \\
    -iG(\omega) & F(\omega^{-1})
    \end{array}\right)\\
    &= -i\left(\begin{array}{cc}
    F^2(\omega) - G^2(\omega) & -i\Big[F(\omega)G(\omega^{-1}) + F(\omega^{-1})G(\omega)\Big]\\
    i\Big[F(\omega)G(\omega^{-1}) + F(\omega^{-1})G(\omega)\Big] & -F^2(\omega^{-1}) + G^2(\omega^{-1})
    \end{array}\right).
\end{aligned}
\end{equation}
By Lemma~\ref{lem:even_poly}, we can construct a QSP phase sequence $\vec{\phi}$ such that
\begin{equation}
    \bra{1}{\rm QSP}_{Z,k}(\theta/2)\ket{0} = \frac{1}{k+1} \sum_{l=-k,\,l\text{ even}}^{k} e^{il\theta/2}
\end{equation}
and thus 
\begin{equation}\label{Eq:evenerr_X}
    \left| \bra{1}{\rm QSP}_{Z,k}(\theta/2)\ket{0} \right|^2 = \sum_{l=-k}^{k} \frac{k-|l|+1}{(k+1)^2} e^{il\theta}.
\end{equation}

If the unknown unitary is $U_X(\theta) = e^{-i X \theta}$, we have
\begin{equation}
    {\rm QSP}_{X,k}(\theta/2) \coloneqq R_x(\phi_0) \left( \prod_{j=1}^n R_x(\theta) R_x(\phi_j)\right) R_z(\pi)R_x(-\phi_{n}) \left( \prod_{j=n-1}^{0} R_x(\theta) R_x(-\phi_j) \right).
\end{equation}
Since we choose $\phi_{:n} \coloneqq \sum_{j=0}^{n} \phi_j = -\pi/2$, we further have
\begin{equation}
    {\rm QSP}_{X,k}(\theta/2) = R_x(n\theta + \phi_{:n}) R_z(\pi) R_x(n\theta - \phi_{:n}) = \left(\begin{array}{cc}
        -i\cos \phi_{:n} & \sin\phi_{:n} \\
        -\sin \phi_{:n} & i\cos\phi_{:n}
    \end{array}\right) = \left(\begin{array}{cc}
        0 & -1 \\
        1 & 0
    \end{array}\right).
\end{equation}
It follows that $\left|\bra{0}{\rm QSP}_{X,k}(\theta/2)\ket{0}\right|^2 = 1$. Combining this with Eq.~\eqref{Eq:evenerr_X}, we have
\begin{equation}\label{Eq:XZ_even_suc_geq}
\begin{aligned}
    &\quad\,\, {\rm Suc}_k^{\texttt{SEQ}}\left(\{X, Z\}, \{p,1-p\}, \pi \right) \\
    &\geq \frac{1}{4\pi^2} \int_{0}^{2\pi} \int_{0}^{2\pi} p\left|\bra{0} {\rm QSP}_{Z,k}(\theta_0/2) \ket{0}\right|^2 + (1-p)\left|\bra{0} {\rm QSP}_{X,k}(\theta_1/2) \ket{1}\right|^2 \mathrm{d}\theta_0  \mathrm{d}\theta_1 = \frac{k+p}{k+1}.
\end{aligned}
\end{equation}
Hence, we complete the achievability part.

\begin{remark}
We remark on the robustness of our QSP-based protocol under depolarizing noise. Suppose in practice, every quantum gate is concatenated with a depolarizing noise $\cD_\lambda(\cdot) = (1-\lambda)(\cdot) + \lambda\idop/2$. Since the depolarizing channel commutes with unitary gates, which means that applying a unitary gate followed by a depolarizing channel is equivalent to applying the depolarizing channel first. By calculating the average success probability, we have
\begin{equation}
    {\rm Suc}_{k,\text{depo}}^{\texttt{SEQ}}\left(\{X, Z\}, \{p,1-p\}, \pi \right) = \begin{cases}
        \frac{1}{2} + \frac{1}{2}(1 - \lambda)^{2k+1} \cdot \frac{k}{k + 1} & \text{when } k \text{ is odd} \\
        \frac{1}{2} + \frac{1}{2}(1 - \lambda)^{2k+3} \cdot \frac{k}{k + 1} & \text{when } k \text{ is even}.
    \end{cases}
\end{equation}
\end{remark}

\subsection{Optimality}
The primal problem for calculating the optimal average success probability of discriminating $X$ and $Z$ over time can be written as
\begin{equation}\label{Eqappe:opt_suc_pri_XZ}
{\rm Suc}_k^{\texttt{SEQ}}\left(\{X, Z\}, \{p,1-p\}, \pi \right) =\max_{\{T_0, T_1\}\in \cT^{\texttt{SEQ}}}\; \tr\Big(pT_0 \Omega_{X}^{(k)}(\pi) + (1-p) T_1 \Omega_{Z}^{(k)}(\pi)\Big).
\end{equation}
Recall that Ref.~\cite{Bavaresco_2021} presented a method to obtain a dual problem of the primal problem in Eq.~\eqref{Eqappe:opt_suc_pri_XZ} based on the characterization of dual affine spaces. The dual problem of Eq.~\eqref{Eqappe:opt_suc_pri_XZ} can be written as follows, and Slater's condition holds.
\begin{equation}\label{sdpappendix:XZ}
\begin{aligned}
    {\rm Suc}_k^{\texttt{SEQ}}\left(\{X, Z\}, \{p,1-p\},\pi \right) = \min &\;\; \lambda \\
    {\rm s.t.} &\;\; \Omega_{X}^{(k)}(\pi) \leq \frac{1}{p}\lambda \overline{W},\\
    &\;\; \Omega_{Z}^{(k)}(\pi) \leq \frac{1}{1-p}\lambda\overline{W},\\
    &\;\; \overline{W} \in \overline{\cW}^{\texttt{SEQ}}.
\end{aligned}
\end{equation}
For $k$-slot sequential strategies, the dual affine space $\overline{\cW}^{\texttt{SEQ}}$ is given by the set of Choi operators of $k-1$-slot quantum combs~\cite{Bavaresco_2021}. In the following, we will construct a feasible solution to the above optimization problem as $\lambda_k = \frac{k+p}{k+1}$ and
\begin{equation}
    \overline{W}_k = \frac{p(k+1)}{k+p}\Omega_{X}^{(k)}(\pi) + \frac{(1-p)(k+1)}{k+p}\Omega_{Z}^{(k)}(\pi) - \frac{(1-p)}{k+p}\dketbra{\idop}{\idop}^{\ox k}.
\end{equation}
First, by Lemma~\ref{lem:Ck_geq_id}, we have
\begin{equation}
    \Omega_{X}^{(k)}(\pi) - \frac{1}{k+1}\dketbra{\idop}{\idop}^{\ox k} \geq 0, \quad\text{and}\quad \Omega_{Z}^{(k)}(\pi) - \frac{1}{k+1}\dketbra{\idop}{\idop}^{\ox k} \geq 0.
\end{equation}
It follows that $\overline{W}_k \geq 0$. Moreover, it is easy to see that $\Omega_{X}^{(k)}(\pi)$, $\Omega_{Z}^{(k)}(\pi)$, and $\ketbra{\Phi}{\Phi}^{\ox k}$ are Choi operators of $k-1$-slot quantum combs. It is then easy to see that $\overline{W}_k$ is the Choi operator of a $k-1$-slot quantum comb and thus, $\overline{W}\in \overline{\cW}^{\texttt{SEQ}}$. Second, since $\frac{1}{2} \le p < 1$, it can be verified that
\begin{equation}
\begin{aligned}
    & \frac{1}{p}\lambda_k \overline{W}_k - \Omega_{X}^{(k)}(\pi) =\frac{1-p}{p}\Big( \Omega_{Z}^{(k)}(\pi) - \frac{1}{k+1}\dketbra{\idop}{\idop}^{\ox k} \Big) \geq 0,\\
    & \frac{1}{1-p}\lambda_k\overline{W}_k - \Omega_{Z}^{(k)}(\pi) = \frac{p}{1-p}\Big( \Omega_{X}^{(k)}(\pi) - \frac{1}{k+1}\dketbra{\idop}{\idop}^{\ox k}\Big) + \frac{2p-1}{1-p}\dketbra{\idop}{\idop}^{\ox k} \geq 0.
\end{aligned}
\end{equation}
Therefore, $\{\lambda_0, \overline{W}_k\}$ satisfies all constraints in the SDP~\eqref{sdpappendix:XZ} and is a feasible solution to it. As the SDP~\eqref{sdpappendix:XZ} is a minimization problem, it follows that
\begin{equation}\label{Eq:XZ_suc_leq}
    {\rm Suc}_k^{\texttt{SEQ}}\left(\{X, Z\}, \{p,1-p\}, \pi \right)  \leq \frac{k+p}{k+1}.
\end{equation}
Combining Eq.~\eqref{Eq:XZ_odd_suc_geq}, Eq.~\eqref{Eq:XZ_even_suc_geq} and Eq.~\eqref{Eq:XZ_suc_leq}, we arrive at
\begin{equation}
   {\rm Suc}_k^{\texttt{SEQ}}\left(\{X, Z\}, \{p,1-p\}, \pi \right) = \frac{k+p}{k+1}.
\end{equation}
Moreover, we show that general strategies do not offer any advantage for recognizing ${X,Z}$. This is because the dual affine space $\overline{\cW}^{\texttt{GEN}}$ of $\cW^{\texttt{GEN}}$, the set of all general strategies that encompasses indefinite causal orders, is given by the set of Choi operators of $k$-partite non-signaling channels without the positivity constraint~\cite{Bavaresco_2021}:
\begin{equation}
\begin{aligned}
_{O_1} \overline{W}^{\texttt{GEN}} & = _{I_1 O_1} \overline{W}^{\texttt{GEN}}, \\
_{O_2} \overline{W}^{\texttt{GEN}} & = _{I_2 O_2} \overline{W}^{\texttt{GEN}}, \\
& \cdots \\
_{O_k} \overline{W}^{\texttt{GEN}} & = _{I_k O_k} \overline{W}^{\texttt{GEN}}, \\
\tr \overline{W}^{\mathrm{GEN}} & = d_{I_1} d_{I_2} \cdots d_{I_k}.
\end{aligned}
\end{equation}
It can be checked that $\Omega_{X}^{(k)}(\pi),\Omega_{Z}^{(k)}(\pi)$ and $\dketbra{\idop}{\idop}^{\ox k}$ all satisfy the above conditions, implying that $\overline{W}_k$ also satisfies these conditions. Consequently, $\overline{W}_k$ is a feasible solution for the following SDP for recognizing $\{X,Z\}$ with general strategies.
\begin{equation}
\begin{aligned}
    {\rm Suc}_k^{\texttt{GEN}}\left(\{X, Z\}, \{p,1-p\}, \pi \right) = \min &\;\; \lambda \\
    {\rm s.t.} &\;\; \Omega_{X}^{(k)}(\pi) \leq \frac{1}{p}\lambda \overline{W},\\
    &\;\; \Omega_{Z}^{(k)}(\pi) \leq \frac{1}{1-p}\lambda\overline{W},\\
    &\;\; \overline{W} \in \overline{\cW}^{\texttt{GEN}}.
\end{aligned}
\end{equation}
It follows ${\rm Suc}_k^{\texttt{GEN}}\left(\{X, Z\}, \{p,1-p\}, \pi \right) \leq \frac{k+p}{k+1}$. Since ${\rm Suc}_k^{\texttt{GEN}}\left(\{X, Z\}, \{p,1-p\}, \pi \right) \geq {\rm Suc}_k^{\texttt{SEQ}}\left(\{X, Z\}, \{p,1-p\}, \pi \right)$ by definition, we conclude that ${\rm Suc}_k^{\texttt{GEN}}\left(\{X, Z\}, \{p,1-p\}, \pi \right) = \frac{k+p}{k+1}$.

Now we consider the case for general $\vec{n}_0$ and $\vec{n}_1$. Since $\vec{n}_0\cdot\vec{n}_1 = 0$, we can construct a third orthonormal real vector $\vec{n}_2$ that is orthogonal to $\vec{n}_0$ and $\vec{n}_1$. Consider a rotation matrix $R\in \mathrm{SO}(3)$
\begin{equation}\label{Eq:rot_matR}
    R = \left(\begin{array}{ccc}
       h_{0x} & h_{2x} & h_{1x} \\
       h_{0y} & h_{2y} & h_{1y} \\
       h_{0z} & h_{2z} & h_{1z} \\
    \end{array}\right).
\end{equation}
Note there is a correspondence between rotation matrices in $\mathrm{SO}(3)$ and unitary operators in $\mathrm{SU}(2)$. Specifically, for every rotation matrix $R$, there exists a unitary operator $U$ such that
\begin{equation}~\label{eq:rotate_R_unitary}
    U(\vec{n}\cdot \vec{\sigma})U^{\dagger} = (R\vec{n})\cdot \vec{\sigma},~\forall \vec{n} \in \mathbb{R}^3.
\end{equation}
Therefore, we have $U H_0 U^\dag = U(\vec{n}_0 \cdot \vec{\sigma}) U^\dag = [1, 0, 0]\cdot \vec{\sigma} = X$,
and $U H_1 U^\dag = U(\vec{n}_1 \cdot \vec{\sigma}) U^\dag = [0, 0, 1]\cdot \vec{\sigma} = Z$. It follows that 
\begin{equation}
    \widetilde{U}_0 = Ue^{-i\theta (\vec{n}_0\cdot \vec{\sigma})} U^\dag = e^{-i\theta U(\vec{n}_0\cdot \vec{\sigma})U^\dag} = e^{-i\theta X},
\end{equation}
and
\begin{equation}
    \widetilde{U}_1 = Ue^{-i\theta (\vec{n}_2\cdot \vec{\sigma})} U^\dag = e^{-i\theta U(\vec{n}_1\cdot \vec{\sigma})U^\dag} = e^{-i\theta Z}.
\end{equation}
After pre-processing the unknown evolution and applying the protocol for $\{X,Z\}$, we complete the proof.

\begin{remark}
Second, we note that the variance of the recognition success probability can be calculated as 
\begin{equation}
\begin{aligned}
    & \frac{1}{4\pi^2}\int_0^{2\pi} \int_0^{2\pi} \Bigg[p\left|\bra{0} {\rm QSP}_{Z,k}(\theta_0/2) \ket{0}\right|^2 + (1-p)\left|\bra{1} {\rm QSP}_{X,k}(\theta_1/2) \ket{1}\right|^2 - \frac{k+p}{k+1}\Bigg]^2 {\rm d}\theta_0 {\rm d}\theta_1 \\
    =&\; \frac{1}{4\pi^2}\int^{2\pi}_0\int^{2\pi}_0 \Bigg[ p +  \frac{(1-p)k}{k+1} - \frac{2(1-p)}{(k+1)^2}\sum_{l=1}^{k}(k-l+1)\cos(l\theta_1) - \frac{k+p}{k+1} \Bigg]^2 \mathrm{d}\theta_0 \mathrm{d}\theta_1 \\
    =&\;  \frac{1}{2\pi}\int^{2\pi}_0 \Bigg[ - \frac{2(1-p)}{(k+1)^2}\sum_{l=1}^{k}(k-l+1)\cos(l\theta_1)\Bigg]^2 \mathrm{d}\theta_1 \\
    =&\; \frac{1}{2\pi} \frac{2(1-p)}{(k+1)^4} \sum_{i=1}^k \sum_{j=1}^k \frac{(k-i+1)(k-j+1)}{2} \int^{2\pi}_0 \Bigg[\cos \Big((i+j)\theta_1\Big)+\cos\Big(\big(i-j)\theta_1\Big)\Bigg] \mathrm{d}\theta_1\\
    =&\; \frac{k(2k+1)(1-p)}{6(k+1)^3},
\end{aligned}
\end{equation}
where for  $(\sum_{l=1}^{k}(k-l+1)\cos(l\theta_1))^2$, we have
\begin{equation}
\begin{aligned}
    \big(\sum_{l=1}^{k}(k-l+1)\cos(l\theta_1)\big)^2 &= \sum_{i=1}^k\sum_{j=1}^k \frac{(k-i+1)(k-j+1)}{2} \Bigg[\cos\Big((i+j)\theta_1\Big)+\cos\Big(\big(i-j)\theta_1\Big)\Bigg].
\end{aligned}
\end{equation}
We see that the variance of the recognition success probability vanishes very quickly with the increasing number of queries to the unknown unitary.
\end{remark}

\section{Proof of Theorem 3}\label{appendix:proof_XYZ}
\renewcommand\theproposition{\textcolor{blue}{3}}
\begin{theorem}[Ternary Hamiltonian recognition]~\label{thm:recog_XYZ}
For $H_0 = \vec{n}_0\cdot \vec{\sigma},~H_1 = \vec{n}_1\cdot \vec{\sigma}$ and $H_2 = \vec{n}_2\cdot \vec{\sigma}$ with prior probabilities $p_0,p_1,p_2$, respectively, where $p_0+p_1+p_2=1$, if $\vec{n}_0 \perp \vec{n}_1 \perp \vec{n}_2$, the optimal average success probability of recognizing $\{H_0, H_1, H_2\}$ via a $k$-slot sequential strategy, where $k$ is odd, is given by
\begin{equation}
    {\rm Suc}_k^{\texttt{SEQ}}\left(\big\{H_i\big\}_{i=0}^2, \{p_i\}_{i=0}^2, \pi \right) = \frac{k+\max\{p_0,p_1,p_2\}}{k+1}.
\end{equation}
\end{theorem}
In the following, the proof is presented for the range $p_2 \ge p_1 \ge p_0$ as well.
\renewcommand{\theproposition}{S\arabic{proposition}}

\subsection{Achievability}
We will first prove the case $H_0=X,H_1=Y,H_2=Z$. The main idea is to coherently use the protocols given in Theorem 1 for recognizing $\{X,Y\}$ and $\{Y,Z\}$. Recall that by Theorem 1, there exists a QSP phase sequence $\vec{\phi}_x = (\phi_{x,0},\phi_{x,1},\cdots,\phi_{x,k})$ such that one can use the following QSP circuit to recognize $\{X,Z\}$
\begin{equation}
    {\rm QSP}^{\{X,Z\}}_{H,k}(\theta) = R_z(\phi_{x,0}) \prod_{j=1}^{k} U_H(\theta) R_z(\phi_{x,j}),
\end{equation}
where $U_H(\theta) = e^{-i H \theta/2}$ is an unknown unitary operation with $H\in\{X,Z\}$. After applying ${\rm QSP}^{\{X,Z\}}_{H,k}(\theta)$ on $\ket{0}$ and measuring the output state in the computational basis, we decide $H=X$ if the outcome is `1'. Theorem 1 ensures that this protocol has a success probability $\frac{2k+1}{2k+2}$ by querying the unknown unitary $k$ times. In specific, by Theorem 1, we have 
\begin{equation}
    {\rm QSP}^{\{X,Z\}}_{X,k}(\theta/2) = \begin{pmatrix}
        P_{x,k}(\theta) & -Q_{x,k}^*(\theta) \\
        Q_{x,k}(\theta) & P^*_{x,k}(\theta)
    \end{pmatrix},
\end{equation}
with
\begin{equation}
    \int_0^{2\pi} |Q_{x,k}(\theta)|^2 {\rm d}\theta = \frac{k}{k+1}
\end{equation}
as the average success probability of recognizing $X$. Similarly, there exists a QSP phase sequence $\vec{\phi}_y = (\phi_{y,0},\phi_{y,1},\cdots,\phi_{y,k})$ such that 
\begin{equation}
    {\rm QSP}^{\{Y,Z\}}_{H,k}(\theta) = R_z(\phi_{y,0}) \prod_{j=1}^{k} U_H(\theta) R_z(\phi_{y,j}).
\end{equation}
This QSP circuit can be used to optimally sense $Y$ and $Z$, and we have ${\rm QSP}^{\{Y,Z\}}_{Y,k}(\theta)\ket{0} = P_{y,k}(\theta) \ket{0} + Q_{y,k}(\theta) \ket{1}$. In particular, we could set $Q_{y,k}(\theta) = iQ_{x,k}^*(\theta)$. Based on the above, we construct a protocol for recognizing $\{X, Y, Z\}$. The quantum circuit of the protocol is given as follows.
\begin{equation*}
    \Qcircuit @C=0.65em @R=1.0em{
    \lstick{\ket{0}} &\gate{H} & \ctrl{1} & \ctrl{2} & \qw &\ctrl{2} &\qw &\ctrl{2} &\qw &\ctrl{2} &\qw &\cdots & &\qw &\ctrl{2} &\qw &\ctrl{2} &\ctrl{2} \barrier[-2.85em]{2} & \gate{H} &\targ  &\qw &\meter & \rstick{s_0} \cw\\ 
    \lstick{\ket{0}} & \qw & \targ & \qw & \gate{R_z(\phi_{x,k})} &\qswap &\qw &\qswap &\gate{R_z(\phi_{x,k-1})} &\qswap &\qw &\cdots & &\qw &\qswap &\gate{R_z(\phi_{x,0})} &\qswap &\qw  &\gate{U(\boldsymbol{\phi}_x, \boldsymbol{\phi}_y)} &\ctrl{-1} &\qw & \qw & \qw \\
    \lstick{\ket{0}} & \qw & \qw & \targ & \gate{R_z(\phi_{y,k})} &\qswap &\gate{U_H(\theta)} &\qswap  &\gate{R_z(\phi_{y,k-1})} &\qswap &\qw &\cdots & &\gate{U_H(\theta)} &\qswap &\gate{R_z(\phi_{y,0})} &\qswap &\targ &\qw &\ctrl{-2}  &\qw & \meter & \rstick{s_1} \cw\\
    & & & & & & & & & & & & & & & & & \qquad \ket{\psi^{(H)}}
    }
\end{equation*}
Here the angles for rotation gates in the second qubit are $\vec{\phi}_x = (\phi_{x,0},\phi_{x,1},\cdots,\phi_{x,k})$ and the angles for rotation gates in the third qubit are $\vec{\phi}_y = (\phi_{y,0},\phi_{y,1},\cdots,\phi_{y,k})$. Denote $\boldsymbol{\phi}_l =\sum_j \phi_{l,j},~l\in\{x,y\}$. The gate $U(\boldsymbol{\phi}_x, \boldsymbol{\phi}_y)$ has a matrix form
\begin{equation}
U(\boldsymbol{\phi}_x, \boldsymbol{\phi}_y) = \left(
\begin{array}{cc}
   \frac{1}{\sqrt{2}}e^{i\boldsymbol{\phi}_y/2}  & \frac{1}{\sqrt{2}} e^{-i\boldsymbol{\phi}_x/2} \\
   \frac{1}{\sqrt{2}}e^{i\boldsymbol{\phi}_x/2}  & -\frac{1}{\sqrt{2}} e^{-i\boldsymbol{\phi}_y/2}
\end{array}\right).
\end{equation}
Then the protocol is as follows.
\begin{enumerate}
\setlength{\itemsep}{0pt}
\setlength{\parskip}{0pt}
\item Input a zero state $\ket{000}$ and apply the unitary operations.
    \item Measure the first and third qubits to obtain a measurement outcome $s = s_0 s_1$.
    \begin{enumerate}
    \setlength{\itemsep}{0pt}
    \setlength{\parskip}{0pt}
        \item [-] If $s = 00 \textrm{ or } 10$, decide the Hamiltonian is $Z$;
        \item [-] If $s = 01$, decide the Hamiltonian is $Y$;
        \item [-] If $s = 11$, decide the Hamiltonian is $X$.
    \end{enumerate}
\end{enumerate}
In the following, we analyze the specific average success probability of recognizing different unknown Hamiltonians using this protocol. For convenience, we denote the state after the CNOT gate as \begin{equation}
\begin{aligned}
    \ket{\psi^{(H)}} &= \frac{1}{\sqrt{2}} \left( \ket{0}\ox R_z(\boldsymbol{\phi}_x) \ket{0} \ox {\rm QSP}^{\{Y,Z\}}_{H,k} (\theta)\ket{0} + \ket{1} \ox R_z(\boldsymbol{\phi}_y) \ket{1} \ox X \cdot {\rm QSP}^{\{X,Z\}}_{H,k}(\theta)\ket{1} \right)\\
    &= \frac{1}{\sqrt{2}} \left( e^{-i \boldsymbol{\phi}_x / 2} \ket{00} \ox {\rm QSP}^{\{Y,Z\}}_{H,k} (\theta)\ket{0} + e^{i \boldsymbol{\phi}_y / 2} \ket{11} \ox X \cdot {\rm QSP}^{\{X,Z\}}_{H,k}(\theta)\ket{1} \right),
\end{aligned}
\end{equation}
and notice that all gates after the CNOT gate will not affect the third qubit.

\vspace{1.5em}
\underline{If $U_H(\theta) = e^{-i Z\theta}$}, we have
\begin{equation}
    \ket{\psi^{(H)}} = \frac{1}{\sqrt{2}} \left(e^{-i \boldsymbol{\phi}_x / 2} \ket{00}\ox R_z(\boldsymbol{\phi}_y + \theta)\ket{0} + e^{i \boldsymbol{\phi}_y / 2} \ket{11} \ox R_z(\boldsymbol{\phi}_x + \theta)\ket{0}\right).
\end{equation}
Measuring the third qubit at the end will obtain the outcome `0' regardless of what $\theta$ is. Thus, the average success probability in this case is given by $\operatorname{Pr}(s_1 = 0|H = Z) = 1$.

\vspace{1.5em}
\underline{If $U_H(\theta) = e^{-i X\theta}$}, we have
\begin{equation}\label{Eq:caseX_psiprime}
\begin{aligned}
    \ket{\psi^{(X)}} &= \frac{1}{\sqrt{2}} \Big( e^{-i \boldsymbol{\phi}_x / 2} \ket{00} \ox {\rm QSP}^{\{Y,Z\}}_{X,k} (\theta)\ket{0} + e^{i \boldsymbol{\phi}_y / 2} \ket{11} \ox X \cdot {\rm QSP}^{\{X,Z\}}_{X,k}(\theta)\ket{1} \Big)\\
    &= \frac{e^{-i\boldsymbol{\phi}_x/2}}{\sqrt{2}}\ket{00}\ox \big(P_{y,k}(\theta) \ket{0} - i Q_{y,k}(\theta) \ket{1}\big) + \frac{e^{i\boldsymbol{\phi}_y/2}}{\sqrt{2}}\ket{11}\ox \big(P_{x,k}^*(\theta) \ket{0} - Q_{x,k}^*(\theta) \ket{1}\big)\\
    & = \left(P_{y,k}(\theta)\frac{e^{-i\boldsymbol{\phi}_x/2}}{\sqrt{2}}\ket{00} + P_{x,k}^*(\theta) \frac{e^{i\boldsymbol{\phi}_y/2}}{\sqrt{2}}\ket{11} \right)\ox \ket{0} \\
    &\qquad + \left(- iQ_{y,k}(\theta) \frac{e^{-i\boldsymbol{\phi}_x/2}}{\sqrt{2}}\ket{00} - Q_{x,k}^*(\theta) \frac{e^{i\boldsymbol{\phi}_y/2}}{\sqrt{2}}\ket{11} \right)\ox \ket{1},
\end{aligned}
\end{equation}
where in the second equality we used the fact that $R_x(\theta) = R_z(-\pi/2)R_y(\theta)R_z(\pi/2),\forall \theta \in [0,2\pi]$, and thus
\begin{equation}\label{Eq:XputinY}
\begin{aligned}
    {\rm QSP}^{\{Y,Z\}}_{X,k}(\theta/2)\ket{0} &= R_z\left(-\frac{\pi}{2}\right) {\rm QSP}^{\{Y,Z\}}_{Y,k}(\theta/2) R_z\left(\frac{\pi}{2}\right)\ket{0}\\
    &= P_{y,k}(\theta) \ket{0} + e^{-i\pi/2}Q_{y,k}(\theta)\ket{1} = P_{y,k}(\theta) \ket{0} - i Q_{y,k}(\theta)\ket{1}.
\end{aligned}
\end{equation}
It can be seen that if a unitary $e^{-iX\theta}$ is inserted into a QSP circuit ${\rm QSP}^{Y}_{H}$, compared with inserting a unitary $e^{-iY\theta}$, a relative phase $-i$ will be accumulated. Since we determine $H=X$ if the outcome on the first and third qubit is `11' and all gates after the CNOT gate will not affect the third qubit, now we could only consider the case where the third qubit outcome is in $\ket{1}$. Followed by Eq.~\eqref{Eq:caseX_psiprime}, applying post-selection and the rest of gates gives
\begin{alignat}{1}
    &\,\, \left( \idop^{\ox 2} \ox \ketbra{1}{1} \right) \ket{\psi^{(X)}} \\
    =&\,\, \left(- iQ_{y,k}(\theta) \frac{e^{-i\boldsymbol{\phi}_x/2}}{\sqrt{2}}\ket{00} - Q_{x,k}^*(\theta) \frac{e^{i\boldsymbol{\phi}_y/2}}{\sqrt{2}}\ket{11} \right)\ox \ket{1}\\
    \xrightarrow{H} &\,\, Q_{x,k}^* (\theta) \left[\ket{0} \ox \left(\frac{e^{-i\boldsymbol{\phi}_x/2}}{2} \ket{0} - \frac{e^{i\boldsymbol{\phi}_y/2}}{2}\ket{1} \right)\ox \ket{1} + \ket{1} \ox \left(\frac{e^{-i\boldsymbol{\phi}_x/2}}{2} \ket{0} + \frac{e^{i\boldsymbol{\phi}_y/2}}{2}\ket{1} \right)\ox \ket{1}\right]\label{Eq:psi_Had_XZ}\\
    \xrightarrow{U(\boldsymbol{\phi}_x, \boldsymbol{\phi}_y)} &\,\, Q_{x,k}^*(\theta) \left(\frac{1}{\sqrt{2}} \ket{011} + \frac{e^{-i(\boldsymbol{\phi}_x - \boldsymbol{\phi}_y)/2}}{\sqrt{2}} \ket{101}\right)\label{Eq:psi_U_XZ}\\
    \xrightarrow{\operatorname{Toffoli}} &\,\, Q_{x,k}^*(\theta) \left(\ket{1} \ox \frac{1}{\sqrt{2}}\left(\ket{1} + e^{-i(\boldsymbol{\phi}_x - \boldsymbol{\phi}_y)/2} \ket{0}\right) \ox \ket{1}\right),
\end{alignat}
where in Eq.~\eqref{Eq:psi_Had_XZ} we used $Q_{y,k}(\theta) = iQ_{x,k}^*(\theta)$, and in Eq.~\eqref{Eq:psi_U_XZ} we used
\begin{equation}
    U(\boldsymbol{\phi}_x, \boldsymbol{\phi}_y)\left(\frac{e^{-i\boldsymbol{\phi}_x/2}}{2} \ket{0} - \frac{e^{i\boldsymbol{\phi}_y/2}}{2}\ket{1}\right) = \frac{1}{\sqrt{2}}\ket{1},
\end{equation}
and
\begin{equation}
    U(\boldsymbol{\phi}_x, \boldsymbol{\phi}_y)\left(\frac{e^{-i\boldsymbol{\phi}_x/2}}{2} \ket{0} + \frac{e^{i\boldsymbol{\phi}_y/2}}{2}\ket{1} \right) = \frac{e^{-i(\boldsymbol{\phi}_x - \boldsymbol{\phi}_y)/2}}{\sqrt{2}}\ket{0}.
\end{equation}
Therefore, the average probability of getting outcome `11' in the first and third qubits when asserting $e^{-iX\theta/2}$ is determined by 
\begin{equation}\label{Eq:suc_prob_X11}
    \operatorname{Pr}(s = 11|H = X) = \int_{0}^{2\pi} |Q_{x,k}^*(\theta)|^2 {\rm d}\theta = \frac{k}{k+1}.
\end{equation}

\underline{When $U_H(\theta) = e^{-i Y\theta}$}, we have
\begin{equation}
\begin{aligned}
    \ket{\psi^{(Y)}} &= \frac{1}{\sqrt{2}} \Big( e^{-i \boldsymbol{\phi}_x / 2} \ket{00} \ox {\rm QSP}^{\{Y,Z\}}_{Y,k} (\theta/2)\ket{0} + e^{i \boldsymbol{\phi}_y / 2} \ket{11} \ox X \cdot {\rm QSP}^{\{X,Z\}}_{Y,k}(\theta/2)\ket{1} \Big)\\
    &= \frac{e^{-i\boldsymbol{\phi}_x/2}}{\sqrt{2}}\ket{00}\ox \big(P_{y,k}(\theta) \ket{0} + Q_{y,k}(\theta) \ket{1}\big) + \frac{e^{i\boldsymbol{\phi}_y/2}}{\sqrt{2}}\ket{11}\ox \big(P_{x,k}^*(\theta) \ket{0} + i Q_{x,k}^*(\theta) \ket{1}\big)\\
    & = \left(P_{y,k}(\theta)\frac{e^{-i\boldsymbol{\phi}_x/2}}{\sqrt{2}}\ket{00} + P_{x,k}^*(\theta) \frac{e^{i\boldsymbol{\phi}_y/2}}{\sqrt{2}}\ket{11} \right)\ox \ket{0}\\
    &\qquad + \left( Q_{y,k}(\theta) \frac{e^{-i\boldsymbol{\phi}_x/2}}{\sqrt{2}}\ket{00} + iQ_{x,k}^*(\theta) \frac{e^{i\boldsymbol{\phi}_y /2}}{\sqrt{2}}\ket{11} \right)\ox \ket{1},
\end{aligned}
\end{equation}
where in the second equality we used the fact that $R_y(\theta) = R_z(\pi/2)R_x(\theta)R_z(-\pi/2),\forall \theta \in [0,2\pi]$, and thus
\begin{equation}
\begin{aligned}
    {\rm QSP}^{\{X,Z\}}_{Y,k}(\theta/2)\ket{1} &= R_z\left(\frac{\pi}{2}\right) {\rm QSP}^{\{X,Z\}}_{X,k}(\theta/2) R_z\left(-\frac{\pi}{2}\right) \ket{1}\\
    &= - e^{-i\pi/2}Q_{x,k}^*(\theta) \ket{0} + P_{x,k}^*(\theta)\ket{1} = iQ_{x,k}^*(\theta) \ket{0} + P_{x,k}^*(\theta)\ket{1}.
\end{aligned}
\end{equation}
Similarly to the previous one, consider the case where the third qubit is in $\ket{1}$.
\begin{alignat}{1}
    &\,\, \left( \idop^{\ox 2} \ox \ketbra{1}{1} \right) \ket{\psi^{(Y)}} \\
    =&\,\, \left( Q_{y,k}(\theta) \frac{e^{-i\boldsymbol{\phi}_x/2}}{\sqrt{2}}\ket{00} + iQ_{x,k}^*(\theta) \frac{e^{i\boldsymbol{\phi}_y /2}}{\sqrt{2}}\ket{11} \right)\ox \ket{1}\\
    \xrightarrow{H} &\,\, iQ_{x,k}^* (\theta) \left[\ket{0} \ox \left(\frac{e^{-i\boldsymbol{\phi}_x/2}}{2} \ket{0} + \frac{e^{i\boldsymbol{\phi}_y/2}}{2}\ket{1} \right)\ox \ket{1} + \ket{1} \ox \left(\frac{e^{-i\boldsymbol{\phi}_x/2}}{2} \ket{0} - \frac{e^{i\boldsymbol{\phi}_y/2}}{2}\ket{1} \right)\ox \ket{1}\right]\label{Eq:psi_Had_YZ}\\
    \xrightarrow{U(\boldsymbol{\phi}_x, \boldsymbol{\phi}_y)} &\,\, iQ_{x,k}^*(\theta)  \left( \frac{e^{-i(\boldsymbol{\phi}_x - \boldsymbol{\phi}_y)/2}}{\sqrt{2}}\ket{001} + \frac{1}{\sqrt{2}} \ket{111}\right) \\
    \xrightarrow{\operatorname{Toffoli}} &\,\, iQ_{x,k}^*(\theta) \left(\ket{0} \ox \frac{1}{\sqrt{2}}\left( e^{-i(\boldsymbol{\phi}_x - \boldsymbol{\phi}_y)/2} \ket{0} + \ket{1}\right) \ox \ket{1}\right),
\end{alignat}
where we used $-iQ_{y,k}(\theta) = Q_{x,k}^*(\theta)$ in Eq.~\eqref{Eq:psi_Had_XZ}. Therefore, the average probability of getting outcome `01' in the first and third qubits when asserting $e^{-iY\theta}$ is given by 
\begin{equation}\label{Eq:suc_prob_Y01}
    \operatorname{Pr}(s=01|H=Y) = \int_{0}^{2\pi} |iQ_{x,k}^*(\theta)|^2 {\rm d}\theta = \frac{k}{k+1}.
\end{equation}
Combining the above three cases, we have the overall average success probability that satisfies
\begin{equation}\label{Eq:XYZ_suc_geq}
\begin{aligned}
   {\rm Suc}_k^{\texttt{SEQ}}\left(\big\{H_0, H_1, H_2\big\}, \{p_0,p_1,p_2\}, \pi \right) &\geq p_2 \operatorname{Pr}(s_1 = 0|H = Z) + p_0\operatorname{Pr}(s = 11|H = X) + p_1\operatorname{Pr}(s = 01|H = Y) \\
    &= p_2 + \frac{(p_0 + p_1)k}{k+1} = \frac{k+p_2}{k+1}.
\end{aligned}
\end{equation}

\subsection{Optimality}

Similar to the case of binary Hamiltonian recognition, the optimal average success probability can be written as the following SDP problem.
\begin{equation}\label{Eq:opt_suc_pri_XYZ}
    {\rm Suc}_k^{\texttt{SEQ}}\left(\big\{H_0, H_1, H_2\big\}, \{p_0,p_1,p_2\}, \pi \right) = \max_{\substack{\{T_0, T_1, T_2\}\in \cT^{\texttt{SEQ}} \\ p_2 \ge p_1 \ge p_0}}\; \tr\Big(p_0T_0 \Omega_{X}^{(k)}(\pi) + p_1T_1 \Omega_{Y}^{(k)}(\pi) + p_2T_2 \Omega_{Z}^{(k)}(\pi)\Big).
\end{equation}
And the dual problem of Eq.~\eqref{Eq:opt_suc_pri_XYZ} can be written as follows.
\begin{equation}\label{sdp:XYZ}
\begin{aligned}
    {\rm Suc}_k^{\texttt{SEQ}}\left(\big\{H_0, H_1, H_2\big\}, \{p_0,p_1,p_2\}, \pi \right) = \min &\;\; \lambda \\
    {\rm s.t.} &\;\; \Omega_{X}^{(k)}(\pi) \leq \frac{1}{p_0}\lambda \overline{W},\\
    &\;\; \Omega_{Y}^{(k)}(\pi) \leq \frac{1}{p_1}\lambda\overline{W},\\
    &\;\; \Omega_{Z}^{(k)}(\pi) \leq \frac{1}{p_2}\lambda\overline{W},\\
    &\;\; \overline{W} \in \overline{\cW}^{\texttt{SEQ}}.
\end{aligned}
\end{equation}

In the following, we will construct a feasible solution to the above optimization problem as $\lambda_k = \frac{k+p_2}{k+1}$ and
\begin{equation}
    \overline{W}_k = \frac{p_0(1+k)}{k+p_2}\Omega_{X}^{(k)}(\pi) + \frac{p_1(1+k)}{k+p_2}\Omega_{Y}^{(k)}(\pi) + \frac{p_2(1+k)}{k+p_2}\Omega_{Z}^{(k)}(\pi) - \frac{(p_0+p_1)}{k+p_2}\dketbra{\idop}{\idop}^{\ox k}.
\end{equation}
By Lemma~\ref{lem:Ck_geq_id}, we have $\overline{W}_k \geq 0$. It is easy to see that $\Omega_{X}^{(k)}(\pi)$, $\Omega_{Z}^{(k)}(\pi)$, and $\ketbra{\idop}{\idop}^{\ox k}$ are Choi operators of $k-1$-slot quantum combs. Therefore, $\overline{W}_k$ is also the Choi operator of a $k-1$-slot quantum comb, thus $\overline{W}_k \in \overline{\cW}^{\texttt{SEQ}}$. Moreover, by Lemma~\ref{lem:Ck_geq_id}, we have
\begin{equation}
\begin{aligned}
    & \frac{1}{p_0}\lambda_k \overline{W}_k - \Omega_{X}^{(k)}(\pi) = \frac{p_1}{p_0}\Omega_{Y}^{(k)}(\pi) + \frac{p_2}{p_0}\Omega_{Z}^{(k)}(\pi) - \frac{(p_0+p_1)}{p0(k+1)}\dketbra{\idop}{\idop}^{\ox k} \geq 0,\\
    & \frac{1}{p_1}\lambda_k \overline{W}_k - \Omega_{Y}^{(k)}(\pi) =  \frac{p_0}{p_1}\Omega_{X}^{(k)}(\pi) +  \frac{p_2}{p_1}\Omega_{Z}^{(k)}(\pi) - \frac{(p_0+p_1)}{p1(k+1)}\dketbra{\idop}{\idop}^{\ox k} \geq 0,\\
    & \frac{1}{p_2}\lambda_k \overline{W}_k - \Omega_{Z}^{(k)}(\pi) = \frac{p_0}{p_2}\Omega_{X}^{(k)}(\pi) + \frac{p_1}{p_2}\Omega_{Y}^{(k)}(\pi) - \frac{(p_0+p_1)}{p2(k+1)}\dketbra{\idop}{\idop}^{\ox k} \geq 0.
\end{aligned}
\end{equation}
It follows that $\{\lambda_k, \overline{W}_k\}$ satisfies all constraints in the SDP~\eqref{sdp:XYZ} and is a feasible solution to it. Therefore, we have
\begin{equation}\label{Eq:XYZ_suc_leq}
    {\rm Suc}_k^{\texttt{SEQ}}\left(\big\{H_0, H_1, H_2\big\}, \{p_0,p_1,p_2\}, \pi \right) \leq \frac{k+p_2}{k+1}.
\end{equation}
Combining Eq.~\eqref{Eq:XYZ_suc_geq} and Eq.~\eqref{Eq:XYZ_suc_leq}, we arrive at
\begin{equation}
    {\rm Suc}_k^{\texttt{SEQ}}\left(\big\{H_0, H_1, H_2\big\}, \{p_0,p_1,p_2\}, \pi \right) = \frac{k+p_2}{k+1}.
\end{equation}
Using the same argument in Theorem 1, and considering a rotation matrix in Eq.~\eqref{Eq:rot_matR}, we know that the above result holds for any $H_0 = \vec{n}_0\cdot \vec{\sigma},~H_1 = \vec{n}_1\cdot \vec{\sigma}$ and $H_2 = \vec{n}_2\cdot \vec{\sigma}$ such that $\vec{n}_0 \perp \vec{n}_1 \perp \vec{n}_2$. Hence, we complete the proof.

\begin{remark}
We would like to emphasize that during the proof, we have constructed two QSP phase sequences for recognizing Hamiltonian pairs $\{X,Z\}$ and $\{Y,Z\}$, respectively. However, it is noteworthy that every QSP phase sequence that can sense $\{X,Z\}$ can also sense $\{Y,Z\}$, but we could not use the same QSP in our circuit structure. The reason is as follows: if ${\rm QSP}^{\{Y,Z\}}_{H,k}(\theta/2) = {\rm QSP}^{\{X,Z\}}_{H,k}(\theta/2)$, followed by Eq.~\eqref{Eq:XputinY}, we have $P_{y,k}(\theta) = P_{x,k}(\theta), -iQ_{y,k}(\theta) = Q_{x,k}(\theta)$. Thus, the state in Eq.~\eqref{Eq:psi_Had_XZ} can be written as
\begin{equation}
    \ket{0} \ox \left(Q_{x,k}(\theta) \frac{e^{-i\boldsymbol{\phi}_x/2}}{2} \ket{0} - Q_{x,k}^*(\theta) \frac{e^{i\boldsymbol{\phi}_y/2}}{2}\ket{1} \right)\ox \ket{1} + \ket{1} \ox \left(Q_{x,k}(\theta) \frac{e^{-i\boldsymbol{\phi}_x/2}}{2} \ket{0} + Q_{x,k}^*(\theta) \frac{e^{i\boldsymbol{\phi}_y/2}}{2}\ket{1} \right)\ox \ket{1}.
\end{equation}
Notice that we can not extract the same phase for the above state to allow the intermediate state to go to $\ket{0}$ or $\ket{1}$.
\end{remark}
\begin{remark}
Similar to the binary case, general strategies do not offer any advantage in recognizing $\{X,Y,Z\}$.
\end{remark}

\section{Experiments on a superconducting quantum computer}\label{appendix:realdevice_exp}

The QPU device is a surface-13 tunable coupling superconducting quantum processor~\cite{Bu2024}. For the processor, the median single-qubit gate fidelity shows a tiny decrease from about 99.95\% for the isolated Randomized Benchmarking to 99.94\% for the simultaneous one. First, we uniformly sample 64 values of $\theta_j$ from $(0,\pi]$ to construct two sets of gates, $\{e^{-iX\theta_j}\}_{j=1}^k$ and $\{e^{-iZ\theta_j}\}_{j=1}^k$, corresponding to the Hamiltonians $X$ and $Z$, respectively. For each gate, as a given unknown unitary, we repeat Algorithm~\ref{alg:binary_protocol} 512 times and record our predictions.

\begin{algorithm}[H]\label{alg:binary_protocol}
\SetKwInOut{Input}{Input}
\SetKwInOut{Output}{Output}
\Input{$k$ copies of an evolution dynamic $U_{H}(\theta) = e^{-iH \theta}$, where $k$ is odd, $t \in \RR$ is unknown and $H$ is an unknown Hamiltonian from a known set $\{X, Z\}$.}
\Output{ a guess of $H$ from $\{X, Z\}$.}
Determine the vector of angles $\vec{\phi} = \phi_0, \ldots, \phi_k$ such that
\begin{equation}
    \bra{0} R_z(\phi_0) \prod_{j=1}^k W_x(a) R_z(\phi_j) \ket{0} 
    = \frac{2}{k+1} \sum_{l=1,\,l\text{ odd}}^{k} T_l(a).
\end{equation}
This can be achieved following the implementation of Theorem 3-5 in~\cite{gilyen_quantum_2018}\;
Construct the single-qubit circuit as follows:
\begin{equation}
    {\rm QSP}_{H,k}(\theta) = R_z(\phi_0) \prod_{j=1}^k U_{H}(\theta) R_z(\phi_j)
.\end{equation}

Input the zero state $\ket{0}$, run the circuit, and perform a computational basis measurement at the end\;

If the measurement outcome is $0$, return $H = Z$; Otherwise return $H = X$.

\caption{Hamiltonian recognition for $\{X,Z\}$}
\label{alg:xxx}
\end{algorithm}

\begin{figure}[t]
    \centering
    \includegraphics[width=.6\linewidth]{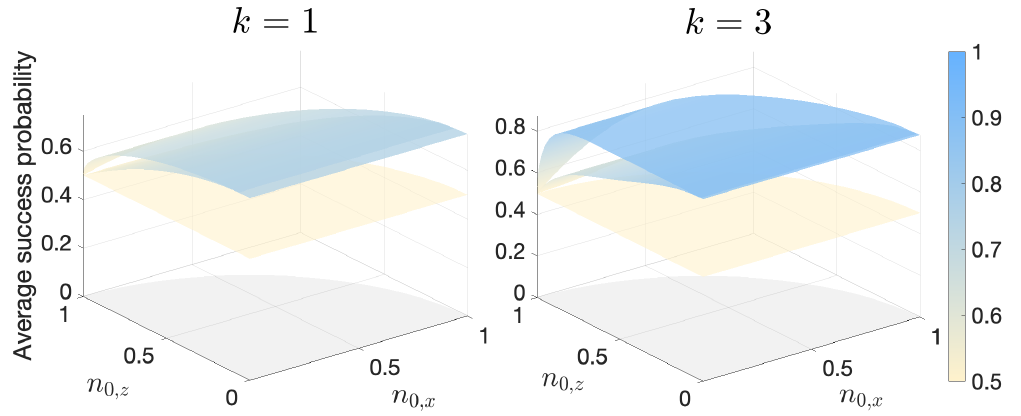}
    \caption{Average success probability of recognizing $\{H_0, Z\}$.}
    \label{fig:generalHZ}
\end{figure}

\section{General axis directions}\label{appendix:general_axis}
Now we consider a more general scenario where $\vec{n}_0$ and $\vec{n}_1$ are not orthogonal for $H_0 = \vec{n}_0\cdot\vec{\sigma}$ and $H_1 = \vec{n}_1\cdot\vec{\sigma}$. Without loss of generality, we set $H_1 = Z$. Since $\vec{n}_0\coloneqq(n_{0,x},n_{0,y},n_{0,z})$ is a unit vector, we can find a unitary $U$ that diagonalize $H_0$ such that $UH_0U^\dagger = Z$. Therefore, for fixed $k$ and $\vec{n}_0$, we could obtain $\Omega_{H_0}^{(k)}(\pi) = (U\ox U^*)\Omega_{Z}^{(k)}(\pi)(U^\dagger \ox U^T)$. Further, we numerically solve SDP
\begin{equation}\label{sdp:XZ}
\begin{aligned}
    {\rm Suc}_k^{\texttt{SEQ}}\big(\{H_0, Z\}, \pi \big) = \min &\;\; \lambda \\
    {\rm s.t.} &\;\; \Omega_{H_0}^{(k)}(\pi) \leq 2\lambda \overline{W},\\
    &\;\; \Omega_{Z}^{(k)}(\pi) \leq 2\lambda\overline{W},\\
    &\;\; \overline{W} \in \overline{\cW}^{\texttt{SEQ}}.
\end{aligned}
\end{equation}
to estimate ${\rm Suc}_k^{\texttt{SEQ}}\left(\big\{H_0, Z\big\}, \pi \right)$ for varying $n_{0,x}$ and $n_{0,z}$.

We illustrate the average success probability of recognizing $\{H_0, Z\}$ in Fig.~\ref{fig:generalHZs}, for the cases of $k=1$ and $k=3$. In each subplot, the top layer represents the average success probability corresponding to the optimal protocol obtained by solving the SDP; the middle layer depicts the average success probability given by the protocol in Theorem 1 specifically designed for $\{X, Z\}$, and the bottom layer shows the performance of random guessing. It can be observed that within a certain range, the $\{X, Z\}$ protocol performs well. As the number of slots increases, the optimal protocol exhibits high accuracy over a larger range.


\end{document}